\documentclass[conference]{IEEEtran}
\IEEEoverridecommandlockouts
\usepackage{mathrsfs}
\usepackage{amsmath}
\usepackage{yhmath}
\usepackage{graphicx}
\usepackage[]{cite}
\usepackage{amssymb}
\usepackage{tabularx}
\usepackage{bbm}
\usepackage{amsfonts}
\usepackage{bm}
\usepackage{amsthm}
\usepackage[usenames]{color}
\usepackage{algorithmic,algorithm}
\newcommand\relphantom[1]{\mathrel{\phantom{#1}}}

\newtheorem{lemma}{Lemma}
\newtheorem{theorem}{Theorem}

\begin{document}
\bibliographystyle{plain}

\title{Robust Beamforming Design for Sum Secrecy Rate Optimization in MU-MISO Networks}

\author{Pu Zhao, Meng Zhang, Hui Yu, Hanwen Luo and Wen Chen, \textit{Senior Member}, \textit{IEEE}
\thanks{This paper is partially sponsored by National Key Project of China (No.2013ZX03001007-004), by National 973 project (No. 2012CB316106), and by National Natural Science Foundation of China under Grants (No.61471236 and no. 61328101).

The authots are with the Department of Electronic Engineering, Shanghai Jiao Tong University, Shanghai 200240, China (email:\{zhaopu90, mengzhang, yuhui, hwluo. wenchen\}@sjtu.edu.cn), Wen Chen is also with School of Electronics and Automation, Guilin University of Electronic Technology, China.
} }

\maketitle

\begin{abstract}
This paper studies the beamforming design problem of a multi-user downlink network, assuming imperfect channel state information known to the base station. In this scenario, the base station is equipped with multiple antennas, and each user is wiretapped by a specific eavesdropper where each user or eavesdropper is equipped with one antenna. It is supposed that the base station employs transmit beamforming with a given requirement on sum transmitting power. The objective is to maximize the sum secrecy rate of the network. Due to the uncertainty of the channel, it is difficult to calculate the exact sum secrecy rate of the system. Thus, the maximum of lower bound of sum secrecy rate is considered. The optimization of the lower bound of sum secrecy rate still makes the considered beamforming design problem difficult to handle. To solve this problem, a beamforming design scheme is proposed to transform the original problem into a convex approximation problem, by employing semidefinite relaxation and first-order approximation technique based on Taylor expansion. Besides, with the advantage of low complexity, a zero-forcing based beamforming method is presented in the case that base station is able to nullify the eavesdroppers' rate. When the base station doesn't have the ability, user selection algorithm would be in use. Numerical results show that the former strategy achieves better performance than the latter one, which is mainly due to the ability of optimizing beamforming direction, and both outperform the signal-to-leakage-and-noise ratio based algorithm.
\end{abstract}

\section{Introduction}

Although the initial  concept about secrecy transmission can be traced back to the 1970s [1], wireless transmission issues concerning physical layer security rate have attracted considerable attention in recent years. Traditional communication methods based on High-level encryption can hardly be used to improve physical layer security rate in practical situation, such as WLAN and Ad-hoc network. In WLAN scenario, the unpredictable random access and leave of users would lead to difficulties for the establishment of an appropriate and reasonable high-level encryption protocol. In addition, in Ad-hoc networks, a complete data transmission could go through several hops and other users may participate in relaying data, which would decrease secrecy rate.

One basic idea on the physical layer security is artificial noise. It adds artificial noise to the transmission signal expecting that the artificial noise would provide more negative effects for eavesdroppers than legitimate users. In [2] and [3], the authors analyze the effect of artificial noise on enhancing secrecy capacity. The difference is that in [2] the transmitter generates the artificial noise, while in [3] an external node helps to accomplish this work. In paper [4], the secure transmission in multiuser and  multiple-eavesdropper  systems is investigated. Several strategies are illustrated to suppress the channel interference. Numerical results demonstrate that adding artificial noise to the transmission message would improve the system secrecy rate.
In paper [5], the authors investigate secure communication between two multi-antenna nodes with an undetected eavesdropper. Unlike previous work, no information regarding the eavesdropper is available. Artificial interference is applied to mask the desired signal. The authors maximize the power available to hide the desired signal from a potential eavesdropper, while maintaining a prespecified signal-to-interference-plus-noise-ratio (SINR) at the desired receiver. The case of the presence of imperfect channel state information (CSI) is also studied.

Another category to protect the physical layer security is based on beamforming. In [6], a relay assisted system is studied with the assumption that the relay is unreliable. Precoding designs of the base station (BS) and relay are provided to maximize the secrecy capacity. In addition, [7] takes the two-way relay scenario into consideration. In this situation, eavesdroppers would utilize the received information obtained from two transmission slots to get the desired eavesdropped users' information. Optimal precoding design strategy in this scenario is presented. In [8], a joint beamforming design of the source and relay based on quality-of-service (QoS) requirements with
presence of channel uncertainty is investigated.
[9] studies the secrecy rate in decode-and-forward relay scenario with finite-alphabet input. A power control scheme based on semidefinite programming is presented for the purpose of maximizing secrecy rate with finite-alphabet input.
[10] investigates the relay-eavesdropper network with two models of imperfect knowledge of the eavesdropper¡¯s channel. The approximation of the ergodic secrecy rate is studied under the Rician fading channel model and the worst-case secrecy rate is considered under the deterministic uncertainty model. Under both models, the optimal rank-1, match-and-forward (MF), and zero-forcing (ZF) beamformers are developed. The effectiveness of the proposed relay beamformers are verified by the numerical results. Paper [11] provides precoding strategies in a coordinated multi-point (CoMP) transmission system.



Robust design of beamforming and artificial noise has been investigated in multiple-input-single-output (MISO) networks. In [12], the authors address the physical layer security in MISO communication systems. The transmission covariance matrices of the steering information and the artificial noise are investigated to maximize the worst-case secrecy rate in a resource-constrained system and to minimize the use of resources to ensure an average secrecy rate. [12] investigates  a three node network including only one user.
Paper [13] considers a multiuser MISO downlink system with the presence of passive eavesdroppers and potential eavesdroppers. The problem of minimizing the total transmit power takes into account artificial noise and energy signal generation for protecting the transmitted information against both considered types of eavesdroppers. The semi-definite programming (SDP) relaxation approach is adopted to obtain the optimal solution. Both [12] and [13] use deterministic model for modelling the CSI uncertainty. Paper [14] proposes a linear precoder for a multiuser MIMO system in which multiusers potentially act as eavesdroppers. The proposed precoder is based on regularized channel inversion with a regularization parameter and power allocation vector. Then, an extension of the algorithm by jointly optimizing the regularization parameter and the power allocation vector is presented to maximize the secrecy sum-rate. Robust beamforming design in relay systems has also been investigated. [15] investigates the non-robust and robust cases of joint optimization in bidirectional multi-user multi-relay MIMO systems. The authors mainly concentrate on the sum MSE criterion as well as maximum user's MSE. In [16], the relaying robust beamforming for device-to-device communication with channel uncertainty is considered.

In our paper, beamforming design of sum secrecy rate (SSR) optimization is investigated under sum power constraint for multiuser MISO channel. To solve the original complicated and nonconvex problem, an efficient approximation algorithm based on Taylor expansion is developed for the purpose of near-optimal solutions. In addition, a beamforming design scheme based on ZF with low complexity is presented. Numerical results demonstrate the better performance of the former algorithm. Both of the algorithms are compared to the SLNR algorithm, which mainly minimizes the power leaking to other user¡¯s channel space, and proved to be better.

The remaining part of this paper is organized as follows. The system model and the sum power constrained beamforming problem are presented in Section II. The solution based on Taylor expansion of the approximation problem under the assumption of imperfect CSI is discussed in detail in Section III. In Section IV,
the beamforming design scheme based on ZF at the BS is demonstrated. Simulation results are presented in Section V. Finally, the conclusions are drawn in Section VI.

$\textit{Notation}$: In this paper, we use bold uppercase and lowercase letters to denote matrices and vectors, respectively. $(\cdot)^{T}$ and $(\cdot)^{H}$ denote the transpose and the conjugate transpose of a matrix or a vector, respectively. $(\cdot)^{*}$ denotes the conjugate of a matrix or a vector which means $(\cdot)^{*}=((\cdot)^{T})^{H}$. $\textrm{Tr}(\cdot)$ is the trace of a matrix. $\textrm{rank}(\cdot)$ denotes the rank of a matrix. We use the expression $x\sim\mathcal{N}(u,\sigma^2)$ if $x$ is complex Gaussian distributed with mean $u$ and variance $\sigma^2$. $\parallel{\cdot}\parallel$ denotes the Frobenius norm. $\succeq$ represents the property of semidefinite. $ x^+ := \textrm{max}\{ x,0 \}$. $C(n,m)$ denotes the number of all combinations of $m$ different elements chosen from $n$ elements.

\section{System Model and Problem Formulation}

\subsection{System Model}

In this paper, we will investigate a network incorporating one BS and 2$K$ users, or equivalently $K$ user-eaves pairs. In the user-eaves pair, one legitimate user is wiretapped by an eavesdropper as shown in Fig. 1. It is presumed that the BS is equipped with $N_t$ antennas where $N_t\geq K$, and each user or eavesdropper with a single antenna. It is supposed that the BS serves the $K$ wiretapped users while each eavesdropper attempts to wiretap the legitimate user in the same user-eaves pair. It is assumed that the BS only knows imperfect CSI of each user and eavesdropper due to limited feedback or other reasons, and the channel estimation error is norm-bounded. The BS employs transmit beamforming to communicate with $K$ users. Let $s_i(t)$ denote the information signal sent for the $i$-th user at time t, and let $\mathbf{w}_{i}\in\mathbb{C}^{N_t\times1}$ be the corresponding beamforming vector. The received signal at the $i$-th wiretapped user is given by
\begin{eqnarray}\label{eq:1}
x_i(t) = \mathbf{h}_{i}^T \mathbf{w}_{i}s_i(t) + \sum_{k = 1,k\neq i}^K \mathbf{h}_{i}^T \mathbf{w}_{k}s_k(t) + n_i(t),
\end{eqnarray}
where $\mathbf{h}_{i} \in\mathbb{C}^{N_t\times 1}$  denotes the channel vector between the BS and the $i$-th user, and ${n}_i$ is additive Gaussian noise at the $i$-th user satisfying $n_i\sim\mathcal{N}(0,\sigma_i^2)$. As seen from (\ref{eq:1}), each eavesdropped user suffers from the intracell interference in addition to the noise. It is presumed that all receivers employ single-user detection where the intracell interference is simply treated as background noise.

\begin{figure}[!htbp]
\centering
\includegraphics[scale=0.4]{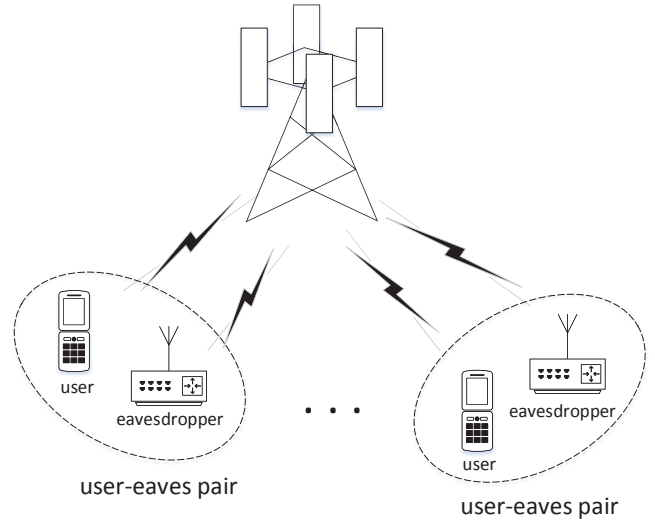}
\caption{Illustration of multiple user-eaves pair wiretapping model}
\label{fig_network_model}
\end{figure}

Under the assumption that each user is wiretapped by the specific eavesdropper, let the $i$-th eavesdropper be the one wiretapping the $i$-th user, which means the $i$-th eavesdropper has the knowledge of $\mathbf{w}_{i}$. The received signal at the $i$-th eavesdropper can be written as
\begin{eqnarray}\label{eq:2}
z_i(t) = \mathbf{g}_{i}^T \mathbf{w}_{i}s_i(t) +
\sum_{k=1,k\neq i}^K \mathbf{g}_{i}^T \mathbf{w}_{k}s_k(t) + m_i(t),
\end{eqnarray}
where $\mathbf{g}_{i} \in\mathbb{C}^{N_t\times 1}$ denotes the channel vector between the BS and the $i$-th eavesdropper, and $m_i$ is additive Gaussian noise at the $i$-th eavesdropper satisfying $m_i\sim\mathcal{N}(0,\varsigma_i^2)$.

The instantaneous achievable rate of the $i$-th user can be transformed into the following
\begin{eqnarray}\label{eq:3}
r_i = \textrm{log}_2\left( 1+ \frac{\vert \mathbf{h}_{i}^T \mathbf{w}_{i} \vert^2 }
{  \sum_{k\neq i} \vert \mathbf{h}_{i}^T \mathbf{w}_{k}\vert^2 + \sigma_i^2 }  \right),
\end{eqnarray}
and the rate achieved by the $i$-th eavesdropper can be expressed as
\begin{eqnarray}\label{eq:4}
s_i = \textrm{log}_2\left( 1+ \frac{\vert \mathbf{g}_{i}^T \mathbf{w}_{i} \vert^2 }
{  \sum_{k\neq i} \vert \mathbf{g}_{i}^T \mathbf{w}_{k}\vert^2 + \varsigma_i^2 }  \right).
\end{eqnarray}

Notice that in this paper, the channel is supposed to be imperfect. Under the assumption that the channel estimation error is norm-bounded, we have
\begin{eqnarray}\label{eq:5}
\mathbf{h}_i=\mathbf{\overline{h}}_i + \triangle \mathbf{h}_i, \|\triangle \mathbf{h}_i\| \leq \varepsilon_h, \\
\mathbf{g}_i=\mathbf{\overline{g}}_i + \triangle \mathbf{g}_i, \|\triangle \mathbf{g}_i\| \leq \varepsilon_g,
\end{eqnarray}
where $\mathbf{\overline{h}}_i$ and $\mathbf{\overline{g}}_i$ are the estimated channel of users and eavesdroppers, while $\triangle \mathbf{h}_i$ and $\triangle \mathbf{g}_i$ are the channel estimation errors, respectively. $\varepsilon_h$ and $\varepsilon_g$ are the bounds of the norm of the channel estimation error. It is presumed that $\varepsilon_h = \varepsilon_g = \varepsilon$.

Based on the expressions above, the SSR is defined as
\begin{eqnarray}\label{eq:6}
\sum_{i=1}^K (r_i - s_i).
\end{eqnarray}

Notice that in this paper, sum power constraint should be satisfied which is
\begin{eqnarray}\label{eq:7}
\qquad  \sum_{i=1}^K { \Vert \mathbf{w}_{i} \Vert^2 \leq P}.
\end{eqnarray}

\subsection{Problem Formulation}

The objective is to maximize the SSR with the constraint of sum transmitting power. Thus, the optimization problem can be formulated as
\begin{subequations}\label{eq:8}
\begin{eqnarray} 
\nonumber
&&{\mathop{\mathop {\textrm{max}}_{\mathbf{w}_{i}\in \mathbb{C}^{N_t\times1}} }_{i=1,\ldots, K} \sum_{i=1}^K \Bigg[ { { \textrm{log}_2} \left( 1+ \frac{\vert \mathbf{h}_{i}^T \mathbf{w}_{i} \vert^2 }{  \sum_{k\neq i} \vert \mathbf{h}_{i}^T \mathbf{w}_{k}\vert^2 + \sigma_i^2 }  \right)} } \\
&&\relphantom{\mathop{\mathop {\textrm{max}}_{\mathbf{w}_{i}\in \mathbb{C}^{N_t}} }_{i=1,\ldots, K}} -{\textrm{log}_2\left( 1+ \frac{\vert \mathbf{g}_{i}^T \mathbf{w}_{i} \vert^2 }{  \sum_{k\neq i} \vert \mathbf{g}_{i}^T \mathbf{w}_{k}\vert^2 + \varsigma_i^2 }  \right)}  \Bigg],\\
&&\quad \textrm{s.t.}\ \quad \sum_{i=1}^K { \Vert \mathbf{w}_{i} \Vert^2 \leq P}.
\end{eqnarray}
\end{subequations}

However, the BS only has the knowledge of the estimated channel and the bound of the norm of the channel estimation error, which would bring difficulties to the modeling of the SSR problem. To deal with this, the lower bound of the objective in (\ref{eq:8}) would be taken into consideration merely. Lemma 1 would be introduced first to get the lower bound of the objective.
\begin{lemma}
: For the two problems in the following,
\begin{eqnarray}\label{eq:9}
\mathop{\textrm{max}}_{\|\mathbf{x}\| \leq \sigma} \phi(\mathbf{x})=\textrm{Re}(\mathbf{x}^{H}\mathbf{y}),\\
\mathop{\textrm{min}}_{\|\mathbf{x}\| \leq \sigma} \varphi(\mathbf{x})=\textrm{Re}(\mathbf{x}^{H}\mathbf{y}),
\end{eqnarray}
where $\sigma$ and $\mathbf{y}$ are given parameters, their solutions can be expressed as
\begin{eqnarray}\label{eq:10}
\phi(\frac{\sigma}{\|\mathbf{y}\|}\mathbf{y} )=\sigma\|\mathbf{y}\|,\\
\varphi(-\frac{\sigma}{\|\mathbf{y}\|}\mathbf{y} )=-\sigma\|\mathbf{y}\|.
\end{eqnarray}
\end{lemma}

\begin{proof}
 Apparently, we have
\begin{eqnarray}\label{eq:39}
-| \mathbf{x}^{H}\mathbf{y}|  \leq \textrm{Re}(\mathbf{x}^{H}\mathbf{y}) \leq  | \mathbf{x}^{H}\mathbf{y} |,
\end{eqnarray}
where $\mathbf{x}$ and $\mathbf{y}$ are vectors. The Cauchy-Schwarz inequality states that for all vectors x and y of an inner product space, it is true that  $|\langle x,y \rangle|^2 \leq  \langle x,x \rangle \cdot \langle y,y \rangle$, where $\langle\cdot,\cdot\rangle$ is the inner product. According to Cauchy-Schwarz inequality,
\begin{eqnarray}\label{eq:40}
|\mathbf{x}^{H}\mathbf{y}|  \leq \|\mathbf{x}\|\|\mathbf{y}\| \leq \sigma\|\mathbf{y}\|.
\end{eqnarray}
Then we have
\begin{eqnarray}\label{eq:41}
-\sigma\|\mathbf{y}\|  \leq \textrm{Re}(\mathbf{x}^{H}\mathbf{y}) \leq  \sigma\|\mathbf{y}\|.
\end{eqnarray}
The inequality holds with equality when $\mathbf{x}$ and $\mathbf{y}$ are linearly dependent.

It can be concluded that the upper bound of $\textrm{Re}(\mathbf{x}^{H}\mathbf{y})$ is $\sigma\|\mathbf{y}\|$ at the point  $\mathbf{x}=\frac{\sigma}{\|\mathbf{y}\|}\mathbf{y} $, and the lower bound of $\textrm{Re}(\mathbf{x}^{H}\mathbf{y})$ is $-\sigma\|\mathbf{y}\|$ at the point $\mathbf{x}=-\frac{\sigma}{\|\mathbf{y}\|}\mathbf{y} $.
\end{proof}

Then, based on Lemma 1, the upper bound of $|\mathbf{h}_i \mathbf{w}_i|^2$ is given by
\begin{eqnarray}\label{eq:11}
\nonumber
&&|\mathbf{h}_i^T \mathbf{w}_i|^2\\\nonumber
&&=\mathbf{h}_i^T \mathbf{w}_i \mathbf{w}^H_i \mathbf{h}^*_i\\\nonumber
&&=\left( \overline{\mathbf{h}}_i^T + \triangle \mathbf{h}_i^T \right) \mathbf{w}_i \mathbf{w}^H_i \left( \overline{\mathbf{h}}^*_i + \triangle \mathbf{h}^*_i \right)\\\nonumber
&&=\mathbf{\overline{h}}_i^T \mathbf{w}_i \mathbf{w}^H_i \mathbf{\overline{h}}^*_i + 2 \textrm{Re} \left\{ \triangle \mathbf{h}_i^T \mathbf{w}_i \mathbf{w}^H_i \mathbf{\overline{h}}^*_i \right\}\\
&&\leq\mathbf{\overline{h}}_i^T \mathbf{w}_i \mathbf{w}^H_i \mathbf{\overline{h}}^*_i + 2 \varepsilon \| \mathbf{w}_i \mathbf{w}^H_i \mathbf{\overline{h}}^*_i\|.
\end{eqnarray}
Similarly, the lower bound of $|\mathbf{h}_i \mathbf{w}_i|^2$ can be expressed as
\begin{eqnarray}\label{eq:12}
\nonumber
&&|\mathbf{h}_i^T \mathbf{w}_i|^2\\\nonumber
&&=\mathbf{\overline{h}}_i^T \mathbf{w}_i \mathbf{w}^H_i \mathbf{\overline{h}}^*_i + 2 \textrm{Re} \left\{ \triangle \mathbf{h}_i^T \mathbf{w}_i \mathbf{w}^H_i \mathbf{\overline{h}}^*_i \right\}\\
&&\geq\mathbf{\overline{h}}_i^T \mathbf{w}_i \mathbf{w}^H_i \mathbf{\overline{h}}^*_i - 2 \varepsilon \| \mathbf{w}_i \mathbf{w}^H_i \mathbf{\overline{h}}^*_i\|.
\end{eqnarray}
Notice that in (\ref{eq:11}) and (\ref{eq:12}), the second order error terms $\triangle \mathbf{h}_i^T \mathbf{w}_i \mathbf{w}^H_i \triangle \mathbf{h}^*_i$ is omitted since it is quite small compared with other terms. In (\ref{eq:12}), $\varepsilon$ is usually small and $\mathbf{\overline{h}}_i^T \mathbf{w}_i \mathbf{w}^H_i \mathbf{\overline{h}}^*_i - 2 \varepsilon \| \mathbf{w}_i \mathbf{w}^H_i \mathbf{\overline{h}}^*_i\|$ is larger than 0. When $\varepsilon$ is large enough, the lower bound, $\mathbf{\overline{h}}_i^T \mathbf{w}_i \mathbf{w}^H_i \mathbf{\overline{h}}^*_i - 2 \varepsilon \| \mathbf{w}_i \mathbf{w}^H_i \mathbf{\overline{h}}^*_i\|$, may be smaller than 0, which has no practical meaning. Due to $\sum_{k=1}^K \mathbf{\overline{h}}_i^T \mathbf{w}_i \mathbf{w}^H_i \mathbf{\overline{h}}^*_i - 2 \varepsilon \| \mathbf{w}_i \mathbf{w}^H_i \mathbf{\overline{h}}^*_i\|$, the problem what is large enough $\varepsilon$ is dependent on the channels and the beamformers. The sum of several terms makes it more difficult to analyses the effect of large $\varepsilon$. From the simulation results, we find that if $\varepsilon \geq 0.6 $, we may not be able to get the beamformers or the sum secrecy rate will be quite low. In this case, the channel estimation error is intolerant and our method may not be able to solve this problem.

Thus, the lower bound of the objective in (\ref{eq:8}) is as follows.
\begin{small}
 \begin{eqnarray}\label{eq:13}
\nonumber
\sum_{i=1}^K \Bigg[ \textrm{log}_2 \left( \frac{ \sum^K_{k=1}(\mathbf{\overline{h}}_i^T \mathbf{w}_k \mathbf{w}^H_k \mathbf{\overline{h}}^*_i - 2 \varepsilon \| \mathbf{w}_k \mathbf{w}^H_k \mathbf{\overline{h}}^*_i\| ) +\sigma^2_i}{  \sum_{k\neq i}( \mathbf{\overline{h}}_i^T \mathbf{w}_k \mathbf{w}^H_k \mathbf{\overline{h}}^*_i + 2 \varepsilon \| \mathbf{w}_k \mathbf{w}^H_k \mathbf{\overline{h}}^*_i\|)+\sigma^2_i  } \right) &&\\\nonumber
\\\nonumber
- \textrm{log}_2\left( \frac{ \sum^K_{k=1}(\mathbf{\overline{g}}_i^T \mathbf{w}_k \mathbf{w}^H_k \mathbf{\overline{g}}^*_i + 2 \varepsilon \| \mathbf{w}_k \mathbf{w}^H_k \mathbf{\overline{g}}^*_i\| ) +\varsigma^2_i}{  \sum_{k\neq i}( \mathbf{\overline{g}}_i^T \mathbf{w}_k \mathbf{w}^H_k \mathbf{\overline{g}}^*_i - 2 \varepsilon \| \mathbf{w}_k \mathbf{w}^H_k \mathbf{\overline{g}}^*_i\|)+\varsigma^2_i  } \right)\Bigg].&&
\\
\end{eqnarray}
\end{small}

Problem (\ref{eq:14}) can be formulated as the following
\begin{small}
\begin{subequations}\label{eq:14}
 \begin{eqnarray}
 \nonumber
 &&\mathop{\mathop {\textrm{max}}_{\mathbf{w}_{i}\in \mathbb{C}^{N_t\times1}} }_{i=1,\ldots, K}  \\ \nonumber
 &&\sum_{i=1}^K \Bigg[ \textrm{log}_2 \left( \frac{ \sum^K_{k=1}(\mathbf{\overline{h}}_i^T \mathbf{w}_k \mathbf{w}^H_k \mathbf{\overline{h}}^*_i - 2 \varepsilon \| \mathbf{w}_k \mathbf{w}^H_k \mathbf{\overline{h}}^*_i\| ) +\sigma^2_i}{  \sum_{k\neq i}( \mathbf{\overline{h}}_i^T \mathbf{w}_k \mathbf{w}^H_k \mathbf{\overline{h}}^*_i + 2 \varepsilon \| \mathbf{w}_k \mathbf{w}^H_k \mathbf{\overline{h}}^*_i\|)+\sigma^2_i  } \right) \\\nonumber
 \\ \nonumber
 && - \textrm{log}_2\left( \frac{ \sum^K_{k=1}(\mathbf{\overline{g}}_i^T \mathbf{w}_k \mathbf{w}^H_k \mathbf{\overline{g}}^*_i + 2 \varepsilon \| \mathbf{w}_k \mathbf{w}^H_k \mathbf{\overline{g}}^*_i\| ) +\varsigma^2_i}{  \sum_{k\neq i}( \mathbf{\overline{g}}_i^T \mathbf{w}_k \mathbf{w}^H_k \mathbf{\overline{g}}^*_i - 2 \varepsilon \| \mathbf{w}_k \mathbf{w}^H_k \mathbf{\overline{g}}^*_i\|)+\varsigma^2_i  } \right)\Bigg],\\
 \\
&&\quad \textrm{s.t.}\ \quad \sum_{i=1}^K { \Vert \mathbf{w}_{i} \Vert^2 \leq P}.
\end{eqnarray}
\end{subequations}
\end{small}

In problem (9), the perfect CSI is not known to the BS. The BS only has the knowledge of the estimated channel and the bound of the norm of the channel estimation error. Under this condition, problem (9) can't be solved since the parameters of the real channels in problem (9)  are missing. Then, the problem to maximize the lower bound of SSR is investigated and problem (20) is formulated based on the estimated channel. If the BS can obtain perfect CSI, that means the bound of the norm of the channel estimation error is set to zero. Then problem (20) would be the same as problem (9).

The upper bound (17) and the lower bound (18) are important in the formulation of problem (20). It is not easy to estimate the tightness of the upper bound and the lower bound. In the simulation results,  the solutions $\mathbf{w}_i$ are obtained. Then we compute the lower bound of SSR through (20a) and the practical SSR through (9a). A comparison between them can illustrate that the gap between the lower bound of SSR and the practical SSR is small.

The traditional leakage-based beamforming scheme [17] might be applicable to such situation, which mainly minimizes the power leaking to other user¡¯s channel space. The solution is given by
$\mathbf{w}_i^o\varpropto max. eigenvector((\delta_i^2\mathbf{I}+\mathbf{\tilde{H}}_i^H \mathbf{\tilde{H}}_i)^{-1}\mathbf{h}_i^*\mathbf{h}_i^T)$, where $\mathbf{\tilde{H}} = [ \mathbf{h}_1^T\cdots \mathbf{h}_{i-1}^T,\mathbf{h}_{i+1}^T\cdots\mathbf{h}_K^T ]$. The norm of $\mathbf{w}_i^o$ is adjusted according to the transmit power.
Besides the signal-to-leakage-and-noise ratio (SLNR) scheme, our proposed algorithms are illustrated in Section III and Section IV. The comparisons between the proposed algorithms and SLNR algorithm shown in section V will demonstrate that the proposed algorithms are better.

\section{Approximation Method Based on Taylor Expansion}%
In this section, an efficient approximation algorithm is developed for the purpose of local-optimal solutions of problem (\ref{eq:14}). Since the objective function (\ref{eq:14}a) is nonconvex and complicated, problem (\ref{eq:14}) is difficult to solve. A convex approximation method based on Taylor expansion will be presented to handle problem (\ref{eq:14}) efficiently in the following.

\subsection{Convex Approximation}

To make the problem more tractable, new matrixs $\mathbf{W}_i=\mathbf{w}_{i} \mathbf{w}_{i}^H$ are introduced. Thus, we have $\textrm{Tr}(\mathbf{W}_i) =\Vert \mathbf{w}_{i} \Vert^2$. It should be noticed that if $\mathbf{W}_i$ instead of $\mathbf{w}_{i}$ is used as the variables to optimize, $\mathbf{W}_i$ has to satisfy $\textrm{rank}(\mathbf{W}_i)=1$, which would violate the problem's convexity. Then, SDR (semidefinite relaximation), a convex optimization based approximation technique, is applied to omit the rank-one constraint. The approximated problem can be transformed into the following
\begin{subequations}\label{eq:15}
 \begin{eqnarray}
 \nonumber
 &&\mathop{\mathop {\textrm{max}}_{\mathbf{W}_{i}\in \mathbb{C}^{N_t\times N_t}} }_{i=1,\ldots, K}  \\ \nonumber
 &&\sum_{i=1}^K \Bigg[ \textrm{log}_2 \left( \frac{ \sum^K_{k=1}(\mathbf{\overline{h}}_i^T \mathbf{W}_k \mathbf{\overline{h}}^H_i - 2 \varepsilon \| \mathbf{W}_k \mathbf{\overline{h}}^*_i\| ) +\sigma^2_i}{  \sum_{k\neq i}( \mathbf{\overline{h}}_i^T \mathbf{W}_k  \mathbf{\overline{h}}^*_i + 2 \varepsilon \| \mathbf{W}_k \mathbf{\overline{h}}^*_i\|)+\sigma^2_i  } \right) \\ \nonumber
 \\ \nonumber
 && - \textrm{log}_2\left( \frac{ \sum^K_{k=1}(\mathbf{\overline{g}}_i^T \mathbf{W}_k  \mathbf{\overline{g}}^*_i + 2 \varepsilon \| \mathbf{W}_k  \mathbf{\overline{g}}^*_i\| ) +\varsigma^2_i}{  \sum_{k\neq i}( \mathbf{\overline{g}}_i^T \mathbf{W}_k   \mathbf{\overline{g}}^*_i - 2 \varepsilon \| \mathbf{W}_k \mathbf{\overline{g}}^*_i\|)+\varsigma^2_i  } \right)\Bigg],\\
 \\
 &&\quad \textrm{s.t.} \ \quad{ \sum_{i=1}^K \textrm{Tr}(\mathbf{W}_i)\leq P},\\
 &&\ \ \relphantom{\textrm{s.t.}}\quad\  \mathbf{W}_i \succeq 0,\quad i=1,\ldots,K.
\end{eqnarray}
\end{subequations}

It should be mentioned that SDR has been widely used in various beamforming design problems. If the solution of the problem satisfies the rank-one constraints, then eigenvalue decomposition would be utilized to obtain the practical optimal solution; otherwise, randomization technique can be applied [18].

Problem (\ref{eq:15}), however, is still nonconvex yet since the objective function (\ref{eq:15}a) is nonconvex. Therefore, further approximations are needed.
Let us consider the following change of variables,
\begin{subequations}\label{eq:16}
\begin{eqnarray}
&&e^{x_i} \triangleq { \sum^K_{k=1}(\mathbf{\overline{h}}_i^T \mathbf{W}_k \mathbf{\overline{h}}^*_i - 2 \varepsilon \| \mathbf{W}_k \mathbf{\overline{h}}^*_i\| ) +\sigma^2_i},   \\
&&e^{y_i} \triangleq {  \sum_{k\neq i}( \mathbf{\overline{h}}_i^T \mathbf{W}_k  \mathbf{\overline{h}}^*_i + 2 \varepsilon \| \mathbf{W}_k \mathbf{\overline{h}}^*_i\|)+\sigma^2_i  },   \\
&&e^{p_i} \triangleq { \sum^K_{k=1}(\mathbf{\overline{g}}_i^T \mathbf{W}_k  \mathbf{\overline{g}}^*_i + 2 \varepsilon \| \mathbf{W}_k  \mathbf{\overline{g}}^*_i\| ) +\varsigma^2_i},    \\
&&e^{q_i} \triangleq {  \sum_{k\neq i}( \mathbf{\overline{g}}_i^T \mathbf{W}_k   \mathbf{\overline{g}}^*_i - 2 \varepsilon \| \mathbf{W}_k \mathbf{\overline{g}}^*_i\|)+\varsigma^2_i  },
\end{eqnarray}
\end{subequations}
for $i=1,\ldots,K$. Note that $|\mathbf{h}_i^T \mathbf{w}_k|^2 \geq 0 $, which means the lower bound and upper bound of $|\mathbf{h}_i^T \mathbf{w}_k|^2 $ should be no less than 0. Then we have
\begin{eqnarray}\label{eq:42}
e^{x_i} \geq \sigma^2_i, \ e^{y_i} \geq \sigma^2_i, \ e^{p_i} \geq \varsigma^2_i, \ and \ e^{q_i} \geq \varsigma^2_i.
\end{eqnarray}
Due to the sum power constraint, $e^{x_i}$, $e^{y_i}$, $e^{p_i}$ and $e^{q_i}$ won't be infinity. It can be observed that $x_i$, $y_i$, $p_i$ and $q_i$ are bounded.

By substituting (\ref{eq:16}) into (\ref{eq:15}a), one can reformulate problem (\ref{eq:15}) as the following
\begin{subequations}\label{eq:17}
\begin{eqnarray}
&&\mathop{\mathop {\textrm{max}}_{ x_i,y_i,p_i,q_i \in\mathbb{R}}}_{\mathbf{W}_{i}\in S,\forall i}  \textrm{log}_2 \prod_{i=1}^K (e^{(x_i-y_i)-(p_i-q_i)}), \quad\ \\\nonumber
&&\quad\ \  \textrm{s.t.}\ \ \ \quad { \sum^K_{k=1}(\mathbf{\overline{h}}_i^T \mathbf{W}_k \mathbf{\overline{h}}^*_i - 2 \varepsilon \| \mathbf{W}_k \mathbf{\overline{h}}^*_i\| ) +\sigma^2_i}\geq e^{x_i} ,\\
\\\nonumber
&&\quad\ \ \ \relphantom{\textrm{s.t.}}\quad\ {  \sum_{k\neq i}( \mathbf{\overline{h}}_i^T \mathbf{W}_k  \mathbf{\overline{h}}^*_i + 2 \varepsilon \| \mathbf{W}_k \mathbf{\overline{h}}^*_i\|)+\sigma^2_i  } \leq e^{y_i},   \\
\\\nonumber
&&\quad\ \ \ \relphantom{\textrm{s.t.}}\quad\ { \sum^K_{k=1}(\mathbf{\overline{g}}_i^T \mathbf{W}_k  \mathbf{\overline{g}}^*_i + 2 \varepsilon \| \mathbf{W}_k  \mathbf{\overline{g}}^*_i\| ) +\varsigma^2_i} \leq e^{p_i} ,    \\
\\\nonumber
&&\quad\ \ \ \relphantom{\textrm{s.t.}}\quad\ {  \sum_{k\neq i}( \mathbf{\overline{g}}_i^T \mathbf{W}_k   \mathbf{\overline{g}}^*_i - 2 \varepsilon \| \mathbf{W}_k \mathbf{\overline{g}}^*_i\|)+\varsigma^2_i  } \geq e^{q_i},\\
\end{eqnarray}
\end{subequations}
where
\begin{eqnarray}\label{eq:18}
\nonumber
&&S \triangleq \Bigg \{ \mathbf{W}_1,\ldots,\mathbf{W}_K\succeq 0 \ \Big \arrowvert\ \\
&&{\mathbf{W}_{i}\in \mathbb{C}^{N_t\times N_t}},\sum_{i=1}^K\textrm{Tr}( \mathbf{W}_i)\leq P \Bigg \}.
\end{eqnarray}

The objective function can be transformed into  ${\sum_{i=1}^K[(x_i-y_i)-(p_i-q_i)]}\textrm{log}_2 e$.
It can be seen that the objective function is convex. Notice that the equalities in (\ref{eq:16}) have been replaced by inequalities as in (\ref{eq:17}b) to (\ref{eq:17}e). It could be verified by the monotonicity of the objective function that all the inequalities in (\ref{eq:17}b) to (\ref{eq:17}e) hold with equalities at the optimal points. To be specific, in the process of solving problem (\ref{eq:17}), if the inequalities (\ref{eq:17}b) or (\ref{eq:17}e) doesn't hold with equality, we can increase $e^{x_i}$ or $e^{q_i}$ until the equality holds. If the inequalities (\ref{eq:17}c) or (\ref{eq:17}d) doesn't hold with equality, we can decrease $e^{y_i}$ or $e^{p_i}$ until the equality holds. In the mean time, the optimal value of problem (\ref{eq:17}) would also be improved. Thus, the inequalities from (\ref{eq:17}b) to (\ref{eq:17}e) would hold with equalities for the final solutions.

For the purpose of maximization of (\ref{eq:17}a), we maximize $e^{x_i}$ and $e^{q_i}$ which are the lower bound of ${ \sum^K_{k=1}(\mathbf{\overline{h}}_i^T \mathbf{W}_k \mathbf{\overline{h}}^*_i - 2 \varepsilon \| \mathbf{W}_k \mathbf{\overline{h}}^*_i\| ) +\sigma^2_i}$ and ${  \sum_{k\neq i}( \mathbf{\overline{g}}_i^T \mathbf{W}_k   \mathbf{\overline{g}}^*_i - 2 \varepsilon \| \mathbf{W}_k \mathbf{\overline{g}}^*_i\|)+\varsigma^2_i  }$, respectively, while minimizing $e^{y_i}$ and $e^{p_i}$ which are the upper bound of ${  \sum_{k\neq i}( \mathbf{\overline{h}}_i^T \mathbf{W}_k  \mathbf{\overline{h}}^*_i + 2 \varepsilon \| \mathbf{W}_k \mathbf{\overline{h}}^*_i\|)+\sigma^2_i  }$ and ${ \sum^K_{k=1}(\mathbf{\overline{g}}_i^T \mathbf{W}_k  \mathbf{\overline{g}}^*_i + 2 \varepsilon \| \mathbf{W}_k  \mathbf{\overline{g}}^*_i\| ) +\varsigma^2_i}$, as can be observed from (\ref{eq:17}b) to (\ref{eq:17}e). Thus, while solving problem (\ref{eq:17}), the lower bound of the numerator of the objective function (\ref{eq:15}a) are maximized and the upper bound of the denominator are minimized leading to the maximization of (\ref{eq:15}). As explained above, it can be seen that problem (\ref{eq:17}) is an appropriate approximation of problem (\ref{eq:18}).

It can be observed that constraints (\ref{eq:17}c) and (\ref{eq:17}d) are nonconvex resulting in difficulties for optimal solution. Let ${(\widetilde{\mathbf{W}}_k,i=1,\ldots,K ) }$ be a feasible point of problem (\ref{eq:17}). Define
\begin{subequations}\label{eq:19}
\begin{eqnarray}
\widetilde y_i \triangleq {\textrm{ln}} \left (   \sum_{k\neq i}\left ( \mathbf{\overline{h}}_i^T \widetilde{\mathbf{W}}_k  \mathbf{\overline{h}}^*_i
+ 2 \varepsilon \| \widetilde{\mathbf{W}}_k \mathbf{\overline{h}}^*_i\|\right)+\sigma^2_i    \right),   \\
\widetilde p_i \triangleq {\textrm{ln}} \left (  \sum^K_{k=1}\left ( \mathbf{\overline{g}}_i^T \widetilde{\mathbf{W}}_k  \mathbf{\overline{g}}^*_i
+ 2 \varepsilon \| \widetilde{\mathbf{W}}_k  \mathbf{\overline{g}}^*_i\| \right ) +\varsigma^2_i \right),
\end{eqnarray}
\end{subequations}
for $i=1,\ldots,K$. Then $\widetilde y_i$ and $\widetilde p_i$ are feasible to problem (\ref{eq:17}). Aiming to make (\ref{eq:17}c) and (\ref{eq:17}d) convex, these constraints are conservatively approximated at the point $(\{\widetilde y_i\}, \{\widetilde p_i\})$ based on Taylor expansion [19]. The Taylor series of a function $f(x)$ that is infinitely differentiable at a number $a$ is the power series $\sum_{n=0}^\infty \frac{f^{(n)}(a)}{n!}(x-a)^n$. Since both of $e^{y_i}$ and $e^{p_i}$ are convex, their first-order Taylor expansion at $\widetilde y_i$ and $\widetilde p_i$ are respectively given by
\begin{eqnarray}\label{eq:20}
e^{\widetilde y_i}(y_i - \widetilde y_i +1) \quad \textrm{and} \quad e^{\widetilde p_i}(p_i - \widetilde p_i +1).
\end{eqnarray}
Consequently, restrictive approximations for (\ref{eq:17}c) and (\ref{eq:17}d) are given by
\begin{subequations}\label{eq:21}
\begin{eqnarray}
\nonumber
{  \sum_{k\neq i}( \mathbf{\overline{h}}_i^T \mathbf{W}_k  \mathbf{\overline{h}}^*_i + 2 \varepsilon \| \mathbf{W}_k \mathbf{\overline{h}}^*_i\|)+\sigma^2_i  }\leq e^{\widetilde y_i}(y_i - \widetilde y_i +1 ), \\
\\\nonumber
{  \sum^K_{k=1}(\mathbf{\overline{g}}_i^T \mathbf{W}_k  \mathbf{\overline{g}}^*_i + 2 \varepsilon \| \mathbf{W}_k  \mathbf{\overline{g}}^*_i\| ) +\varsigma^2_i} \leq e^{\widetilde p_i}(p_i - \widetilde p_i +1 ).\\
\end{eqnarray}
\end{subequations}

Through the first-order Taylor expansion of $e^{y_i}$ and $e^{p_i}$, the original non-linear terms successfully turn out to be linear leading to convex constraints. By replacing (\ref{eq:17}c) and (\ref{eq:17}d) with (\ref{eq:21}a) and (\ref{eq:21}b), respectively, the following approximation of problem (\ref{eq:17}) can be obtained

\begin{subequations}\label{eq:22}
\begin{eqnarray}
&&\mathop{\mathop {\textrm{max}}_{x_i,y_i,p_i,q_i \in\mathbb{R}}}_{\mathbf{W}_{i}\in S,\forall i} \textrm{log}_2 \prod_{i} (e^{(x_i-y_i)-(p_i-q_i)}),\\ \nonumber
&&\quad \ \ \textrm{s.t.}\ \quad \ \ { \sum^K_{k=1}(\mathbf{\overline{h}}_i^T \mathbf{W}_k \mathbf{\overline{h}}^*_i - 2 \varepsilon\| \mathbf{W}_k \mathbf{\overline{h}}^*_i\| ) +\sigma^2_i} \geq e^{x_i} ,   \\
\\\nonumber
&&\qquad \relphantom{\textrm{s.t.}}\quad \ {  \sum_{k\neq i}( \mathbf{\overline{g}}_i^T \mathbf{W}_k   \mathbf{\overline{g}}^*_i - 2 \varepsilon \| \mathbf{W}_k \mathbf{\overline{g}}^*_i\|)+\varsigma^2_i  } \geq e^{q_i},   \\
\\\nonumber
&&\qquad \relphantom{\textrm{s.t.}}\quad \ {  \sum_{k\neq i}( \mathbf{\overline{h}}_i^T \mathbf{W}_k  \mathbf{\overline{h}}^*_i + 2 \varepsilon\| \mathbf{W}_k \mathbf{\overline{h}}^*_i\|)+\sigma^2_i  } \\
&&\qquad \relphantom{\textrm{s.t.}}\quad \ \ \leq e^{ \widetilde{y}_i}(y_i - \widetilde{y}_i +1 ),
\\\nonumber
&&\qquad \relphantom{\textrm{s.t.}}\quad \ {  \sum^K_{k=1}(\mathbf{\overline{g}}_i^T \mathbf{W}_k  \mathbf{\overline{g}}^*_i + 2 \varepsilon \| \mathbf{W}_k  \mathbf{\overline{g}}^*_i\| ) +\varsigma^2_i} \\
&&\qquad \relphantom{\textrm{s.t.}}\quad \ \ \leq e^{\widetilde{p}_i}(p_i - \widetilde{p}_i +1 ).
\end{eqnarray}
\end{subequations}

Problem (\ref{eq:22}) is a convex optimization problem which can be efficiently solved by CVX, a package for specifying and solving convex programs [20],[21].

In summary, the reformulation above consists of two approximation steps: a) the rank relaxation of $\mathbf{w}_{i}\mathbf{w}^H_{i}$ to $\mathbf{W}_{i}$ through
SDR, and b) constraint restrictions of (\ref{eq:17}b) and (\ref{eq:17}e) to (\ref{eq:22}b) and (\ref{eq:22}e). Note that if problem (\ref{eq:22}) yields a rank-one optimal $\left( \mathbf{W}_{1},...,\mathbf{W}_{K}  \right)$, a rank-one beamforming solution can be readily obtained by rank-one decomposition of $\mathbf{W}_{i} = \mathbf{w}_{i}\mathbf{w}^H_{i}$ for $i=1,...,K$. It is then straightforward that this rank-one beamforming solution $\left( \mathbf{w}_{1},...,\mathbf{w}_{K}  \right)$ is also feasible to the original problem (\ref{eq:14}). Otherwise, randomization technique can be applied.

\subsection{Successive Convex Approximation}

Formulation (\ref{eq:22}) is obtained by approximating problem (\ref{eq:14}) at the given feasible point ${( \widetilde {\mathbf{W}}_i,i=1,\ldots,K ) }$, as described in (\ref{eq:21}). This approximation can be further improved by iterative procedure based on the optimal solution obtained through solving (\ref{eq:22}) in the previous approximation. Specifically, in the $(n)$-th iteration, the following convex optimization problem is solved by CVX,
\begin{subequations}\label{eq:36}
\begin{eqnarray}
\nonumber
&& \mathop{\mathop {\textrm{max}}_{x_i,y_i,p_i,q_i \in\mathbb{R}}}_{\mathbf{W}_{i}\in S,\forall i} \textrm{log}_2 \prod_{i} (e^{(x_i-y_i)-(p_i-q_i)}),\\ \nonumber
&& \quad \  \textrm{s.t.}\ \quad \ \ { \sum^K_{k=1}(\mathbf{\overline{h}}_i^T \mathbf{W}_k \mathbf{\overline{h}}^*_i - 2 \varepsilon\| \mathbf{W}_k \mathbf{\overline{h}}^*_i\| ) +\sigma^2_i} \\
&& \quad \ \ \relphantom{\textrm{s.t.}}\ \quad \ \ \geq e^{x_i} , \\\nonumber
&& \qquad \relphantom{\textrm{s.t.}}\quad \ {  \sum_{k\neq i}( \mathbf{\overline{g}}_i^T \mathbf{W}_k   \mathbf{\overline{g}}^*_i - 2 \varepsilon \| \mathbf{W}_k \mathbf{\overline{g}}^*_i\|)+\varsigma^2_i  } \\
&& \quad \ \ \relphantom{\textrm{s.t.}}\ \quad \ \ \geq e^{q_i},\\\nonumber
&& \qquad \relphantom{\textrm{s.t.}}\quad \ {  \sum_{k\neq i}( \mathbf{\overline{h}}_i^T \mathbf{W}_k  \mathbf{\overline{h}}^*_i + 2 \varepsilon\| \mathbf{W}_k \mathbf{\overline{h}}^*_i\|)+\sigma^2_i  } \\
&& \qquad \relphantom{\textrm{s.t.}}\quad \ \ \leq e^{\widetilde{y}_i [n]}(y_i - \widetilde{y}_i [n] +1 ),
\\\nonumber
&& \qquad \relphantom{\textrm{s.t.}}\quad \ {  \sum^K_{k=1}(\mathbf{\overline{g}}_i^T \mathbf{W}_k  \mathbf{\overline{g}}^*_i + 2 \varepsilon \| \mathbf{W}_k  \mathbf{\overline{g}}^*_i\| ) +\varsigma^2_i} \\
&& \qquad \relphantom{\textrm{s.t.}}\quad \ \ \leq e^{\widetilde{p}_i [n]}(p_i - \widetilde{p}_i [n] +1 ).
\end{eqnarray}
\end{subequations}
The optimal solution of (\ref{eq:36}) is denoted as $\{ \widehat{\mathbf{W}}_k[n],\widehat{x}_i[n],\widehat{y}_i[n],\widehat{p}_i[n],\widehat{q}_i[n],i=1,\ldots,K  \}$. Then, through (\ref{eq:35}),
\begin{subequations}\label{eq:35}
\begin{eqnarray}
\nonumber
&& \widetilde y_{i}[n+1] \triangleq {\textrm{ln}} \Big ( \{  \sum_{k\neq i}\big( \mathbf{\overline{h}}_i^T \widehat{\mathbf{W}}_{k}[n] \mathbf{\overline{h}}^*_i \\
&& \relphantom{\widetilde y_{i}[n-1]\triangleq}+ 2 \varepsilon \| \widehat{\mathbf{W}}_{k}[n] \mathbf{\overline{h}}^*_i\|\big)+\sigma^2_i  \}  \Big),   \\\nonumber
&& \widetilde p_{i}[n+1] \triangleq {\textrm{ln}} \Big ( \{  \sum^{K}_{k=1}\big(\mathbf{\overline{g}}_i^T \widehat{\mathbf{W}}_{k}[n]  \mathbf{\overline{g}}^*_i \\
&& \relphantom{\widetilde y_{i}[n-1]\triangleq}+ 2 \varepsilon \| \widehat{\mathbf{W}}_{k} [n] \mathbf{\overline{g}}^*_i\| \big) +\varsigma^2_i\} \Big).
\end{eqnarray}
\end{subequations}
$\widetilde y_i[n+1]$ and $\widetilde p_i[n+1]$ would be used to form the problem in the $(n+1)$-th iteration. Thus, the iterative process continues until it converges.

To get the initial values $\widetilde y_i[1]$ and $\widetilde p_i[1]$, we first generate $\widehat{\mathbf{w}}_i[0]$  randomly, and calculate $\widehat{\mathbf{W}}_i[0]$ by $\widehat{\mathbf{W}}_i[0] = \widehat{\mathbf{w}}_i[0] \widehat{\mathbf{w}}^H_i[0]$. Through (\ref{eq:35}), $\widetilde y_i[1]$ and $\widetilde p_i[1]$ could be obtained.

The randomly generated $\widehat{\mathbf{w}}_i[0]$ should be checked whether it is suitable for the iteration process. The basic idea is that $x_i$, $y_i$, $p_i$ and $q_i$ should be larger than zero. The solutions of the first iteration $\{ \widehat{\mathbf{W}}_k[1],\widehat{x}_i[1],\widehat{y}_i[1],\widehat{p}_i[1],\widehat{q}_i[1]  \}$ can be obtained based on the initial values $\widetilde y_i[1]$ and $\widetilde p_i[1]$. If one of $\widehat{x}_i[1],\widehat{y}_i[1],\widehat{p}_i[1]$ and $\widehat{q}_i[1]$ is negative, that means the randomly generated  $\widehat{\mathbf{w}}_i[0]$ is not appropriate. $\widehat{\mathbf{w}}_i[0]$ should be generated randomly again and the initial values $\widetilde y_i[1]$ and $\widetilde p_i[1]$ should be updated accordingly. The process of checking the initial values should be applied for another time. If $\widehat{x}_i[1],\widehat{y}_i[1],\widehat{p}_i[1]$ and $\widehat{q}_i[1]$ are positive, the iteration process could continue until it converges. notice that in (22a) and (22d), the terms $\mathbf{\overline{h}}_i^T \mathbf{W}_k \mathbf{\overline{h}}^*_i - 2 \varepsilon \| \mathbf{W}_k \mathbf{\overline{h}}^*_i\|$ could influence results of the checking process of the initial values. It is quite easy to generate appropriate initial values if the bound of the norm of the channel estimation error $\varepsilon$ is small. If $\varepsilon$ is large, it might be hard to get suitable initial values.

The proposed successive convex approximation (SCA) algorithm is summarized in Algorithm 1.

\begin{algorithm}
\caption{: SCA algorithm for solving problem in (\ref{eq:22}). } \label{algo_minballd}
\begin{algorithmic}[1]
\STATE \textbf{Given} $\{ \widehat{\mathbf{w}}_{i} , $i=1,...,K$ \} $ that are feasible to (\ref{eq:22}).
\STATE Set $\widehat{\mathbf{W}}_i [0]=\widehat{\mathbf{w}}_{i} \widehat{\mathbf{w}}^H_{i}$ for $i=1,...,K$, and set $n=0$.
\STATE \textbf{repeat}
\STATE \quad \quad obtain $\widetilde{y}_i [n+1]$ and $\widetilde{p}_i [n+1]$ through (\ref{eq:35})
\STATE \quad \quad $n=n+1$,
\STATE \quad \quad solve problem in (\ref{eq:36}) to get the optimal solution \\ \quad \quad $\widehat{\mathbf{W}}_i [n]$
\STATE \quad \quad compute the optimal value of problem in (\ref{eq:22})
\STATE \textbf{until} the stopping criterion is met. 
\STATE obtain $\mathbf{w}^{\ast}_{i}$ by decomposition of $\widehat{\mathbf{W}}_i [n] = (\mathbf{w}^{\ast}_{i} )(\mathbf{w}^{\ast}_{i})^H$ for all $i$ in the case of $\textrm{rank}(\widetilde{\mathbf{W}}_i [n])=1$; otherwise the randomization technology [11] would be utilized to get a rank-one approximation.
\end{algorithmic}
\end{algorithm}

\subsection{Convergence Analysis}
In fact, iterative process ensures monotonic improvement of SSR.

\begin{theorem}
The sequence $\{ \widetilde y_i[n] , \widetilde p_i[n]  \}$, the solutions $\{ \widehat{x}_i[n],\widehat{y}_i[n],\widehat{p}_i[n],\widehat{q}_i[n],i=1,\ldots,K  \}$ generated by Algorithm 1 and the optimal values in the iterations converge.
\end{theorem}
\begin{proof}
Let $\{ \widehat{\mathbf{W}}_i[n],\widehat{x}_i[n],\widehat{y}_i[n],\widehat{p}_i[n],\widehat{q}_i[n],i=1,\ldots,K  \}$ denote the solutions of problem (\ref{eq:36}) in the $n$-th iteration. It should be noticed that $\widetilde y_i[n]$ and $\widetilde p_i[n]$ are feasible to problem (\ref{eq:36}). Due to the form of the objective function (\ref{eq:36}a), the optimal solution $\widehat{y}_i[n]$ and $ \widehat{p}_i[n]$ must satisfy the following,
\begin{eqnarray}\label{eq:37}
\widehat{y}_i[n] \leq \widetilde y_i[n] \quad \textrm{and} \quad \widehat{p}_i[n] \leq \widetilde p_i[n].
\end{eqnarray}
All the inequalities in (\ref{eq:36}b) to (\ref{eq:36}e) would hold with equalities at the optimal points. According to (\ref{eq:36}d), we have

\begin{eqnarray}\label{eq:38}
\nonumber
&& e^{\widetilde y_i[n+1]} = {  \sum_{k\neq i}( \mathbf{\overline{h}}_i^T \widehat{\mathbf{W}}_k[n]  \mathbf{\overline{h}}^*_i + 2 \varepsilon\| \widehat{\mathbf{W}}_k [n] \mathbf{\overline{h}}^*_i\|)+\sigma^2_i  } \\\nonumber
&&  \relphantom{e^{\overline y_i[n+1]} } = e^{\widetilde{y}_i [n]}(\widehat{y}_i[n] - \widetilde{y}_i [n] +1 )  \\
&& \relphantom{e^{\overline y_i[n+1]} } \leq e^{\widehat{y}_i[n]}.
\end{eqnarray}
The inequality in (\ref{eq:38}) holds because $e^{\widehat{y}_i[n]}$ is approximated with its first-order Taylor expansion, $e^{\widetilde{y}_i [n]}(\widehat{y}_i[n] - \widetilde{y}_i [n] +1 )$. Then we have $\widetilde y_i[n+1] \leq \widehat y_i[n]$ and similarly $\widetilde p_i[n+1] \leq \widehat p_i[n]$.

During the iterative process, $\widetilde y_i[n+1] \leq \widetilde y_i[n]$ and $\widetilde p_i[n+1] \leq \widetilde p_i[n]$, which means $\widetilde y_i[n]$ and $\widetilde p_i[n]$ are monotonic. As illustrated in section III-A, the following inequalities are satisfied,
\begin{eqnarray}\label{eq:43}
\infty > \ e^{y_i} \geq \sigma^2_i \quad \textrm{and} \quad \infty > \ e^{p_i} \geq \varsigma^2_i.
\end{eqnarray}
 Hence $\widetilde y_i[n]$ and $\widetilde p_i[n]$ are bounded. It can be concluded that $\widetilde y_i[n]$ and $\widetilde p_i[n]$ would converge. When the iteration index $n$ is large enough, the two problems solved in the $n$-th iteration and in the $(n+1)$-th iteration respectively would be almost the same.

 We have $\widetilde y_i[n+1] \leq \widehat y_i[n] \leq \widetilde y_i[n] $ and similarly $\widetilde p_i[n+1] \leq \widehat p_i[n] \leq  \widetilde p_i[n]$. As $\widetilde y_i[n]$ and $\widetilde p_i[n]$ converge, the solutions $\widehat y_i[n] $ and $\widehat p_i[n]$ would also converge. Since $\widehat y_i[n+1] \leq \widehat y_i[n]  $ and $\widehat p_i[n+1] \leq \widehat p_i[n]$, the solutions $\widehat x_i[n], \widehat y_i[n], \widehat p_i[n], \widehat q_i[n]$ obtained in the n-th iteration would be feasible in the (n+1)-th iteration. Thus, we have $\widehat x_i[n+1] \geq \widehat x_i[n]  $ and $\widehat q_i[n+1] \geq \widehat q_i[n]$ due to the form of the object function. As the result of the limited transmit power, $\widehat x_i[n]$ and $\widehat q_i[n]$ are bounded. It can be concluded that $\widehat x_i[n]$ and $\widehat q_i[n]$ would converge. The solutions $\{ \widehat{x}_i[n],\widehat{y}_i[n],\widehat{p}_i[n],\widehat{q}_i[n],i=1,\ldots,K  \}$ generated by Algorithm 1 converge. As a result, the optimal values in the iterations converge.
\end{proof}

Problem (9) does not consider the effect of the CSI uncertainties. To involve the effect of the CSI uncertainties, the lower bound of SSR is considered and problem (20) is formulated. Since problem (20) is quite complex and nonconvex, it is hard to get the solutions of problem (20). we use approximation method to solve problem (20). However, the performance loss induced by approximation method is still unknown.

In the system model, it is assumed that $N_t\geq K$. The BS with $N_t$ transmit antennas would be able to send at most $N_t$ data streams at one time which means it can communicate with $K$ single antenna users where $N_t\geq K$. If $N_t < K$, the users would not be able to remove the interference or get their own message. It seems that the approximation scheme proposed above can only protect a small number of users at one time in the case of small number of transmit antennas. In the scenario with massive MIMO, the base station is able to be equipped with hundreds of antennas [22]. Then the proposed scheme can protect a lot more users.

\section{Power Allocation Based on ZF Beamforming}
In this section, for the case of imperfect CSI, the SSR's lower bound maximization problem under the assumption of ZF based beamforming method at the BS is studied.

To apply the algorithm based on ZF beamforming, the relationship between $N_t$ and $K$ has to satisfy $N_t\geq 2K$. Thus, according to the relationship of $N_t$ and $K$, two cases would be discussed.
\subsection{Case of $N_t\geq 2K$}
In this case, the BS will be able to provide enough degree of freedom to nullify the eavesdroppers' rate. For the beamforming vector $\mathbf{w}_i$, the beamforming direction is defined as $\frac{\mathbf{w}_i}{\|\mathbf{w}_i\|}$, while the beamforming power is denoted as $\|\mathbf{w}_i\|^2$. We have $P_i=\|\mathbf{w}_i\|^2$, where $P_i$ denotes the power allocated to the beamforming vector $\mathbf{w}_i$. Then, $\mathbf{w}_i$ can be denoted as
\begin{eqnarray}
\mathbf{w}_i = \frac{\mathbf{w}_i}{\|\mathbf{w}_i\|}\sqrt{P_i}.
\end{eqnarray}

Based on the discussion above, the original problem would be divided into two subproblems, design of the beamforming direction and allocation of the beamforming power, respectively.
\subsubsection{Design of the Beamforming Direction}
This part would discuss the beamforming direction design problem.

Let $\overline{\mathbf{H}} \in \mathbb{C}^{N_t\times 2K }$ denote the channel between the BS and the $K$ user-eaves pairs, where $\overline{\mathbf{H} }= \big( \overline{\mathbf{h}}^*_1\quad \ldots\quad \overline{\mathbf{h}}^*_K \quad \overline{\mathbf{g}}^*_1 \quad\ldots\quad \overline{\mathbf{g}}^*_K  \big)$. Applying ZF based beamforming method needs to compute $\overline{\mathbf{H}}^\dag$, which is the Moore-Penrose pseudo-inverse matrix of $\overline{\mathbf{H}}$. $\mathbf{v}_i^T$ is denoted as the $i$-th row of $\overline{\mathbf{H}}^\dag$. Since $\overline{\mathbf{H}} \in \mathbb{C}^{N_t\times 2K }$ and $N_t\geq 2K$, $\overline{\mathbf{H}}^\dag \overline{\mathbf{H}} = \mathbf{I}$ will hold, where $\mathbf{I} \in \mathbb{C}^{2K\times 2K }$ is a identity matrix.

It should be noticed that for $i=1,\ldots,K$ and $j=1,\ldots,K$,
\begin{eqnarray}\label{eq:23}
\mathbf{v}_i^T \overline{\mathbf{h}}^*_j = \left\{ \begin{array}{ll}
1, & i = j,\\
0, & i\neq j,
\end{array} \right.
\end{eqnarray}
and
\begin{eqnarray}\label{eq:24}
\mathbf{v}_i^T \overline{\mathbf{g}}^*_j = 0.
\end{eqnarray}

$\frac{\mathbf{v}_i^*}{ \| \mathbf{v}_i\|}$ is used as the beamforming direction of $\mathbf{w}_i$. Then, the beamforming vector $\mathbf{w}_i $ can be denoted as the following.
\begin{eqnarray}\label{eq:25}
\mathbf{w}_i = \frac{\mathbf{v}_i^*}{\| \mathbf{v}_i \|} \sqrt{P_i}.
\end{eqnarray}

Applying (\ref{eq:23}) and (\ref{eq:24}) in (\ref{eq:14}a), the SSR can be expressed as
\begin{eqnarray}\label{eq:26}
\sum_{i=1}^{K} \textrm{log}_2 \left( {1 + \frac{ \mathbf{\overline{h}}_i^T \mathbf{W}_i \mathbf{\overline{h}}^*_i - 2 \varepsilon \| \mathbf{W}_i \mathbf{\overline{h}}^*_i\|  }{ \sigma_i^2} } \right),
\end{eqnarray}
which means the absolute cancellation of eavesdroppers' rate and intracell interference due to ZF based beamforming method. This mainly results from the orthogonality of $\mathbf{v}_i^T \overline{\mathbf{h}}^*_j$ and $\mathbf{v}_i^T \overline{\mathbf{g}}^*_j$. By substituting (\ref{eq:25}) into (\ref{eq:26}), the optimization problem turns into the following expression
\begin{subequations}\label{eq:27}
\begin{eqnarray}
&&\mathop{\textrm{max}}_{P_i\geq 0,\forall i} \  \sum_{i=1}^{K} \textrm{log}_2 \big( {1 + \frac{  \left( 1-2\varepsilon_h\|\mathbf{v}_i\|   \right )P_i }{ \| \mathbf{v}_i \|^2 \sigma_i^2 } } \big),  \\
&&\ \ \textrm{s.t.} \  \quad \sum_{i=1}^{K} P_i \leq P.
\end{eqnarray}
\end{subequations}

The ZF based beamforming method would prefix the beamforming direction and then allocate power. On the contrast, the origin problem (\ref{eq:14}) would jointly optimize beamforming direction and power allocation.

\subsubsection{Allocation of the Beamforming Power}
After solving the beamforming direction problem, the original problem turns into a power allocation problem, as illustrated by (\ref{eq:27}). The solution of this problem would be demonstrated in this part.

Noticing that the objective function is nonconvex, the approximation method presented in the previous section can be utilized. Specifically, exponential variables would come into use to substitute the nonconvex terms in (\ref{eq:27}), and the problem would be reformulated as follows
\begin{subequations}\label{eq:28}
\begin{eqnarray}
&&\mathop{\textrm{max}}_{P_i\geq 0,\forall i} \ \textrm{log}_2 \prod_{i=1}^{K} e^{z_i}, \\
&&\ \  \textrm{s.t.}\ \quad \big( {1 + \frac{  \left( 1-2\varepsilon_h\|\mathbf{v}_i\|   \right )P_i }{ \| \mathbf{v}_i \|^2 \sigma_i^2 } } \big) \geq e^{z_i},\\
&&\ \  \ \relphantom{\textrm{s.t.}}\quad \sum_{i=1}^{K} P_i \leq P.
\end{eqnarray}
\end{subequations}

Problem (\ref{eq:28}) is convex. The exact numerical solution of problem (\ref{eq:28}) can be found without iteration process through standard convex solvers such as CVX.

In addition to the solution of an optimal problem obtained from standard convex solvers, the closed-form solution of (\ref{eq:27}) could be derived. To be specific, the Lagrangian function as the following should be considered
\begin{eqnarray}\label{eq:29}
\nonumber
&&\mathfrak{L}(\lambda,P_i,i=1,\ldots,K) = \sum_{i=1}^{K} \textrm{log}_2 \big( {1 + \frac{ \left( 1-2\varepsilon_h\|\mathbf{v}_i\|   \right ) P_i }{ \| \mathbf{v}_i \|^2 \sigma_i^2 } } \big) \\
&&\relphantom{\mathfrak{L}(\lambda,P_i,i=1,\ldots,K)=}- \lambda \sum_{i=1}^{K} P_i,
\end{eqnarray}
where $\lambda$ is the Lagrange multiplier.

The KKT condition of this optimal power allocation problem is
\begin{eqnarray}\label{eq:30}
\frac{\partial \mathfrak{L}(\lambda)}{\partial P_i} \left\{ \begin{array}{ll}
=0, & P_i \geq 0,\\ 
\leq0, & P_i = 0.
\end{array} \right.
\end{eqnarray}

After solving (\ref{eq:30}), the power allocation scheme as the following equations (\ref{eq:31}) satisfies the KKT condition, which achieves the maximal SSR.
\begin{eqnarray}\label{eq:31}
P_i = \big( \frac{1}{\lambda} - \frac{ \| \mathbf{v}_i \|^2 \sigma_i^2 }{\left( 1-2\varepsilon_h\|\mathbf{v}_i\|   \right )} \big) ^+.
\end{eqnarray}
It should be noticed that the the Lagrange multiplier, $\lambda$, should satisfy the sum power constraint as explained in (\ref{eq:32}).
\begin{eqnarray}\label{eq:32}
\sum_{i=1}^{K} \big( \frac{1}{\lambda} - \frac{ \| \mathbf{v}_i \|^2 \sigma_i^2 }{\left( 1-2\varepsilon_h\|\mathbf{v}_i\|   \right )} \big) ^+ = P.
\end{eqnarray}

According to (\ref{eq:23}), if the channel between the BS and the $i$-th user is under good condition, then $\|\mathbf{v}_i\|$ usually keeps small, which would lead to large $P_i$ as demonstrated in (\ref{eq:31}). On the other side, the poor channel condition would introduce large $\|\mathbf{v}_i\|$, and as a result, the allocated power would be small or even zero.

It should be noticed that, in the case of $1-2\varepsilon_h\|\mathbf{v}_i\|\leq 0$, due to (\ref{eq:27}), for the purpose of maximizing the SSR, $P_i$ should be set to zero, since any positive value of $P_i$ would lead to a reduction of the objective.

The power allocation process is essentially water-filling scheme. Generally, the better the channel state of a user is, the more power is allocated to the beamforming vector of the user to make full use of better channel. On the contrary, power allocated to the beamforming vectors of users with worse channels will be less or even none.

\subsection{Case of $N_t < 2K$}
In this case, the BS won't be able to provide enough degree of freedom to nullify the eavesdroppers' rate.

It should be observed that to eliminate the eavesdroppers' rate and intracell interference absolutely, the relationship of the BS antenna number $N_t$ and the user-eaves pair number $K$ have to satisfy $N_t \geq 2K$; otherwise (\ref{eq:23}) and (\ref{eq:24}) won't hold and we won't be able to ensure that every eavesdropper's rate keeps zero. If $N_t < 2K$, user selection can be utilized to select $\hat{K}$ users satisfying $N_t = 2\hat{K}$.

$\hat K$ users need to be selected out of $K$ users satisfying $\hat{K} = \frac{N_t}{2}$ in the situation of $N_t < 2K$. Let $U_i,i=1,\ldots,C(K,\hat K)$ denote the user set composed of $\hat K$ users from the $K$ users. Let $R(U_i)$ denote the maximal SSR calculated through solving (\ref{eq:27}) as demonstrated in the case of $N_t \geq 2K$, after the user set $U_i$ is selected. The user set achieving maximum SSR through user selection can be represented as
\begin{eqnarray}\label{eq:33}
U_{opt} = \textrm{arg}\mathop{\textrm{max}}_{ {i=1,\ldots,C(K,\hat K)}} \big\{ R(U_i) \big\}.
\end{eqnarray}

Exhaustive search would be a simple method to get the solution of (\ref{eq:33}). However, the complexity of exhaustive search would be quite high. It is proposed to use a low complexity method as illustrated below.

The user-eaves pair contrast ratio is defined as
\begin{eqnarray}\label{eq:34}
u_i =  \frac{\|\mathbf{h}_i\|^2}{\|\mathbf{g}_i \|^2} .
\end{eqnarray}
The user-eaves pair with large $u_i$ is deemed to be able to acquire high secrecy rate. Based on this assumption, the $K$ user-eaves pairs are arranged in order of $u_i$. Thus, the $\hat K$ users with better $u_i$ would be selected.

\section{Numerical Results and Discussions}
Numerical results are demonstrated in this section so as to verify the effectiveness of our proposed method. We have $ \sigma_i^2 = \varsigma_i^2 = \sigma^2,i=1,\ldots,K $ and $\varepsilon_h = \varepsilon_g = \varepsilon$. For simplicity of expression, we use a vector $[N_t \ K \ \varepsilon]$ to denote the antenna number equipped at the BS, the number of user-eaves pairs served by the BS and the bound of the norm of the channel estimation error. For instance, [4 2 0.1] represents that the BS is equipped with 4 antennas and serves two users each with an eavesdropper while the bound of the norm of the channel estimation error is 0.1. Let $P$ denote the sum transmitting power of all the beamforming vectors. The SNR is defined as $\frac{P}{\sigma^2}$. $N_t$ and $K$ denote the number of BS antennas and user-eaves pairs, respectively.

\subsection{Demonstration of SSR Performance with Perfect CSI}

Firstly, the SSR under different configurations of $N_t$ and $K$ for the case of perfect CSI are demonstrated. It is noticed that to investigate the case of perfect CSI, we only need to set $\varepsilon = 0$. In fact, our approximation is to optimize the lower bound of the practical SSR. In the case of  $\varepsilon = 0$, the practical SSR would be maximized based on our methods. We would first investigate the performance under perfect CSI and get some insights about our methods.

Fig. 2 shows the results of Taylor expansion based method. $N_t$ is fixed to number 16 and 8 while $K$ changes. It can be seen from Fig. 2 that as the sum power grows, the SSR increases. Increasing sum power will lead to a larger feasible set for the problem (\ref{eq:7}), and it will achieve better performance for SSR.

Then, we examine the performance of SSR achieved by our proposed two algorithms and make a comparison to the traditional SLNR method. As presented in Fig. 3 and Fig. 4, in which we set $K$ to 2 and 4, respectively, it can be observed that our two algorithms outperform the existing SLNR based beamforming method. Besides, the Taylor expansion based algorithm achieves better performance than the ZF based algorithm, which demonstrates that the freedom to optimize beamforming direction would offer performance increase on SSR.

From Fig. 3 and Fig. 4, it can be observed that for the same configurations of $N_t$ and $K$, the two curves of SSR vs SNR obtained from Taylor expansion and ZF based methods, respectively, are nearly linear and keep parallel. The steady gap between two curves could illustrate that the ability of optimizing beamforming direction would provide steady or even fixed gain on SSR as sum power increases. It is also noticed that for the same $K$, as $N_t$ enlarges, the gap between these two algorithms decreases rapidly, which indicates that in the situation of large BS antenna number with small number of user-eaves pairs, the two algorithms achieve almost the same performance.

\begin{figure}[!htbp]
\centering
\includegraphics[scale=0.6]{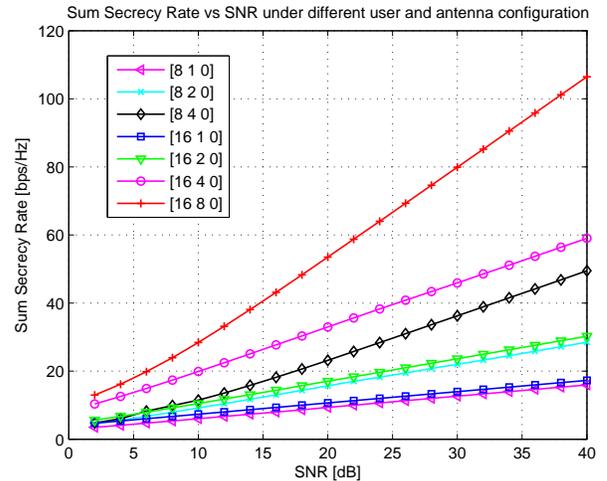}
\caption{Sum Secrecy Rate vs SNR under different user and antenna configuration for Taylor expansion based method.}
\label{fig_FixN16}
\end{figure}

\begin{figure}[!htbp]
\centering
\includegraphics[scale=0.57]{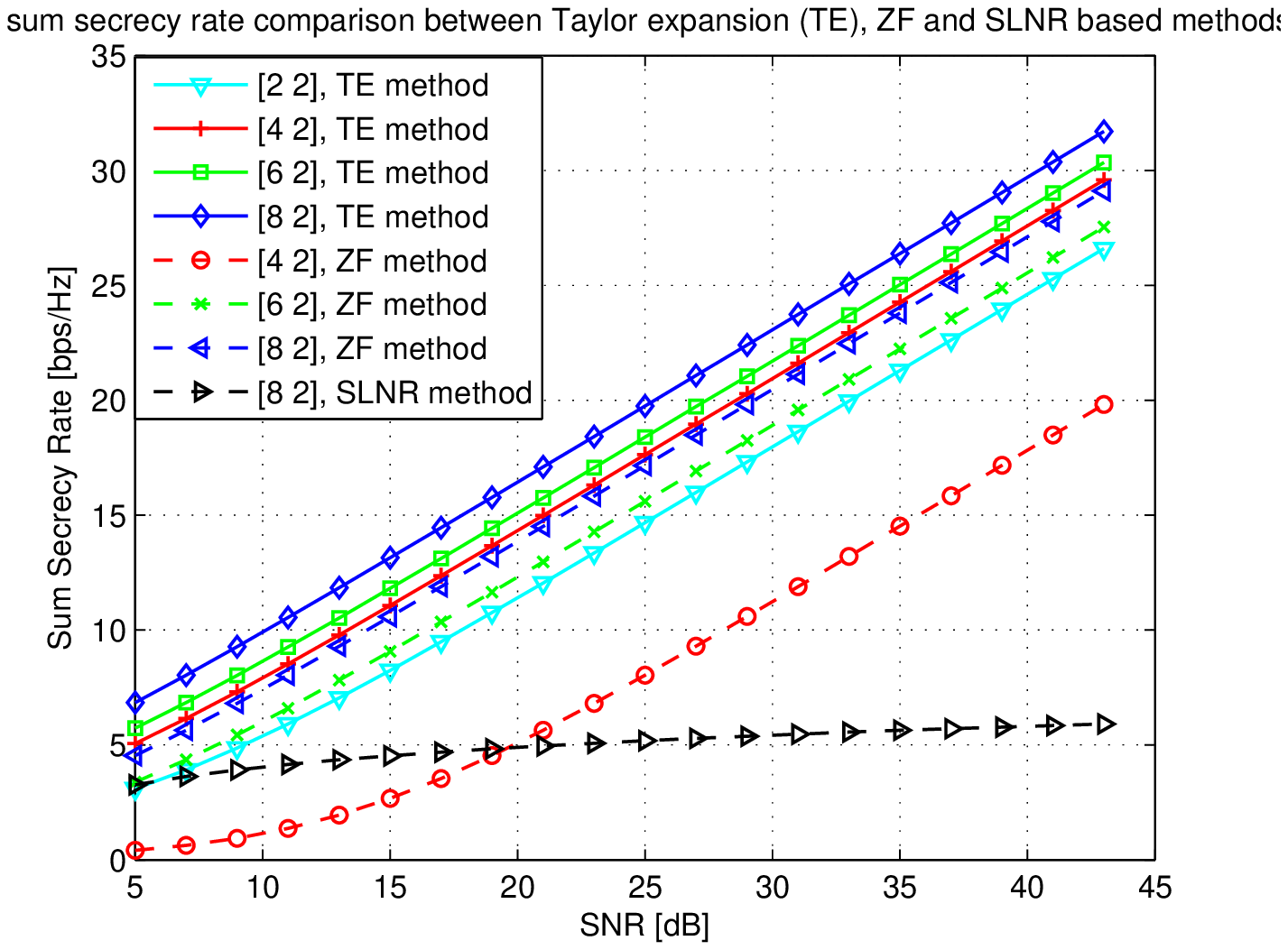}
\caption{sum secrecy rate comparison between Taylor expansion, ZF and SLNR based methods for $K=2$.}
\label{fig_comparison2user}
\end{figure}

\begin{figure}[!htbp]
\centering
\includegraphics[scale=0.57]{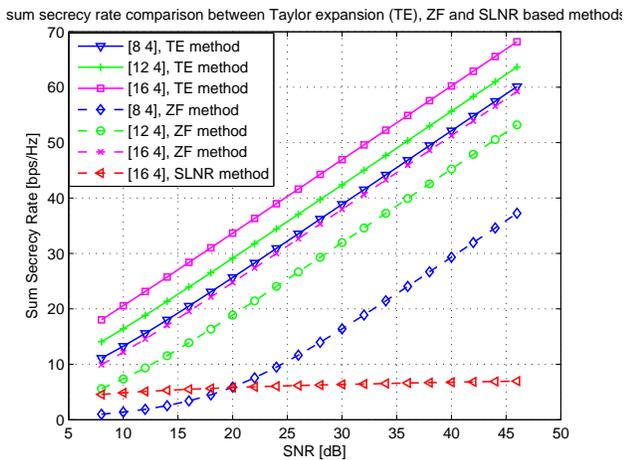}
\caption{sum secrecy rate comparison between Taylor expansion, ZF and SLNR based methods for $K=4$.}
\label{fig_comparison4user}
\end{figure}

\subsection{Demonstration of SSR Performance with Imperfect CSI}

In this case, the performance of SSR under imperfect CSI would be investigated. The results are quite similar to those of the case of perfect CSI. Fig. 5 demonstrates the lower bound of SSR based on Taylor expansion method for $\varepsilon=0.1$. It can be seen from Fig. 5 that as the sum power grows, the lower bound of SSR increases. Increasing sum power will lead to a larger feasible set for the problem (\ref{eq:22}), and it will achieve better performance for the lower bound of SSR.

Notice that in the process of Taylor expansion based approximation method, we get the solutions $\{ \mathbf{W}_i,i=1,\ldots,K  \}$ after solving problem (29). Then we need to compute $\{ \mathbf{w}_i,i=1,\ldots,K  \}$ based on $\{ \mathbf{W}_i,i=1,\ldots,K  \}$. If $ \textrm{rank}(\mathbf{W}_i)=1$, eigenvalue decomposition would be utilized to obtain the practical optimal solution, which is $\mathbf{W}_{i} = \mathbf{w}_{i}\mathbf{w}^H_{i}$ for $i=1,...,K$. However, in most cases, $\{ \mathbf{W}_i,i=1,\ldots,K  \}$ won't satisfy the rank-one condition. Then randomization technique would be applied to get practical solution $\{ \mathbf{w}_i,i=1,\ldots,K  \}$. In Fig. 6, we demonstrate the effect of randomization technique on the lower bound of SSR. In Fig.6, the label "no rand" means that $\{ \mathbf{W}_i,i=1,\ldots,K  \}$ obtained after solving problem (29) is used to calculate the lower bound of SSR directly while the label "rand" denotes that $\{ \mathbf{w}_i,i=1,\ldots,K  \}$ obtained through randomization technique based on $\{ \mathbf{W}_i,i=1,\ldots,K  \}$ is used to compute the lower bound of SSR. As can be observed from Fig. 6, the lower bound of SSR with randomization technique would be lower than the lower bound of SSR without randomization technique. However, the performance gap is quite small and the SDR approximation is effective.

Then, we illustrate the performance of the lower bound of SSR achieved by our proposed two algorithms and make a comparison to the traditional SLNR method. As presented in Fig. 7 and Fig. 8, in which we set $\varepsilon$ to 0.1 and 0.2, respectively, it can be observed that our two algorithms outperform the SLNR based beamforming method. Besides, the Taylor expansion based algorithm achieves better performance than the ZF based algorithm, which demonstrates that the freedom to optimize beamforming direction would offer performance increase on the lower bound of SSR. For the same configurations of $N_t$ and $K$, there is a steady gap between the two curves of the lower bound of SSR vs SNR obtained from Taylor expansion and ZF methods, respectively. The steady gap between two curves could demonstrate that the ability of optimizing beamforming direction would provide steady gain on the lower bound of SSR as sum power increases. It is also noticed that for the same $K$, as $N_t$ enlarges, the gap between these two algorithms decreases rapidly.

Fig. 9 shows the convergence of SCA. As presented in Fig.9, our proposed Taylor expansion based method usually converge after only a few iterations. The algorithm has quick convergence, which is beneficial to obtain the optimal value.

We have demonstrated the performance of the lower bound of SSR. Next, we will illustrate the practical SSR performance based on our methods. The practical SSR is calculated through (\ref{eq:8}a), where the channel is real and the beamforming vectors are obtained by our Taylor expansion based method to maximize the lower bound of SSR. Notice that optimization of the lower bound of SSR is based on the channel with estimation error. In addition, the theoretical SSR can be investigated. The theoretical SSR is still obtained by the optimization method of the lower bound of SSR based on Taylor expansion. However, in theoretical SSR,  the practical channel is used and the estimation error is set to 0. The difference between the theoretical SSR and the lower bound of SSR is the channel and the  estimation error. The theoretical SSR utilizes the real channel while the lower bound of SSR uses the estimated channel, where the estimation error is norm bounded. In theoretical SSR, since the estimation error is set to 0, the method in III would obtain SSR rather than the lower bound of SSR.

As can be observed in Fig. 10, for the same configuration of $N_t$ and $K$, the theoretical SSR is usually larger than the lower bound of SSR and the practical SSR under the same power. The theoretical SSR actually denotes the upper bound of SSR for the channel. The SSR based on the estimated channel should not be larger than it. Also, it can be seen that when the SNR is low, the practical SSR is larger than the lower bound of SSR, which is consistent with our expectations. However, when the SNR is high, the practical SSR would be smaller than the lower bound of SSR. It is because when we approximate the lower bound and the upper bound of $|\mathbf{g}_i \mathbf{w}_i|^2$ through (\ref{eq:11}) and (\ref{eq:12}), the second order error term $\triangle \mathbf{g}_i \mathbf{w}_i \mathbf{w}^H_i \triangle \mathbf{g}^H_i$ is omitted since it is quite small compared with other terms. As the SNR becomes larger, our method is able to keep the lower bound and the upper bound of $|\mathbf{g}_i \mathbf{w}_i|^2$ quite small. Then, it is not appropriate to neglect the influence of the second order error term $\triangle \mathbf{g}_i \mathbf{w}_i \mathbf{w}^H_i \triangle \mathbf{g}^H_i$. As a result of ignoring $\triangle \mathbf{g}_i \mathbf{w}_i \mathbf{w}^H_i \triangle \mathbf{g}^H_i$, the practical SSR would experience some performance loss and be smaller than the lower bound of SSR.

In Fig. 11, the lower bounds of SSR for different bounds of the norm of the channel estimation error are compared. It can be seen that for both Taylor expansion based method and ZF based method, larger bound of the norm of the channel estimation error would decrease the lower bound of SSR. However, the impact is not very significant. For different bounds of the norm of the channel estimation error, the differences of the lower bounds of SSR are little. Larger channel estimation error would surely degrade system performance, but our proposed two methods are able to keep the performance degradation small.

As can be observed in Fig.5, the lower bound of SSR increases with user number k. The normalized lower bound of SSR is defined as the lower bound of SSR divided by user number k. In Fig. 12, the normalized lower bound of SSR for different number of users are compared. It can be seen that as user number k increases, the normalized lower bound of SSR would also increase. That means more users would be beneficial to improve SSR which is mainly due to the proposed beamforming design. It can also be observed that for the same configuration of user number and estimation error norm bound, more transmit antennas would enhance system performance.

\begin{figure}[!htbp]
\centering
\includegraphics[scale=0.57]{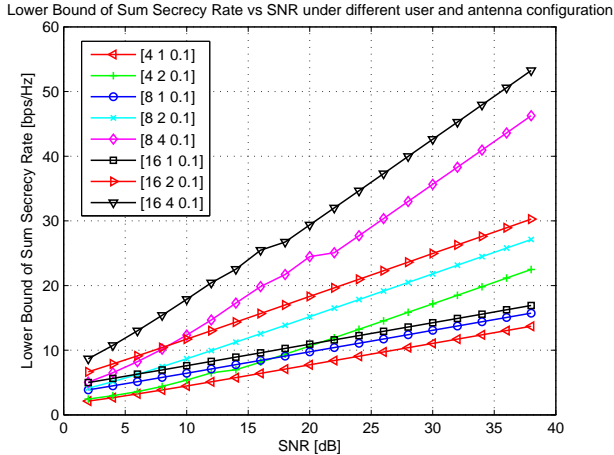}
\caption{lower bound of sum secrecy rate based on Taylor expansion method for $\varepsilon=0.1$.}
\label{fig_TE01}
\end{figure}

\begin{figure}[!htbp]
\centering
\includegraphics[scale=0.57]{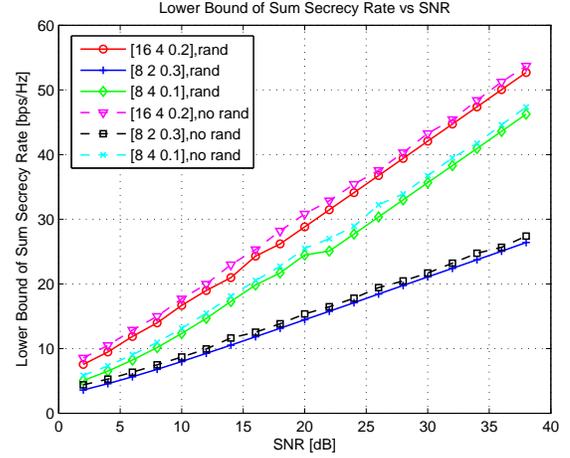}
\caption{the effect of randomization technique in Taylor expansion method.}
\label{fig_randtech}
\end{figure}

\begin{figure}[!htbp]
\centering
\includegraphics[scale=0.57]{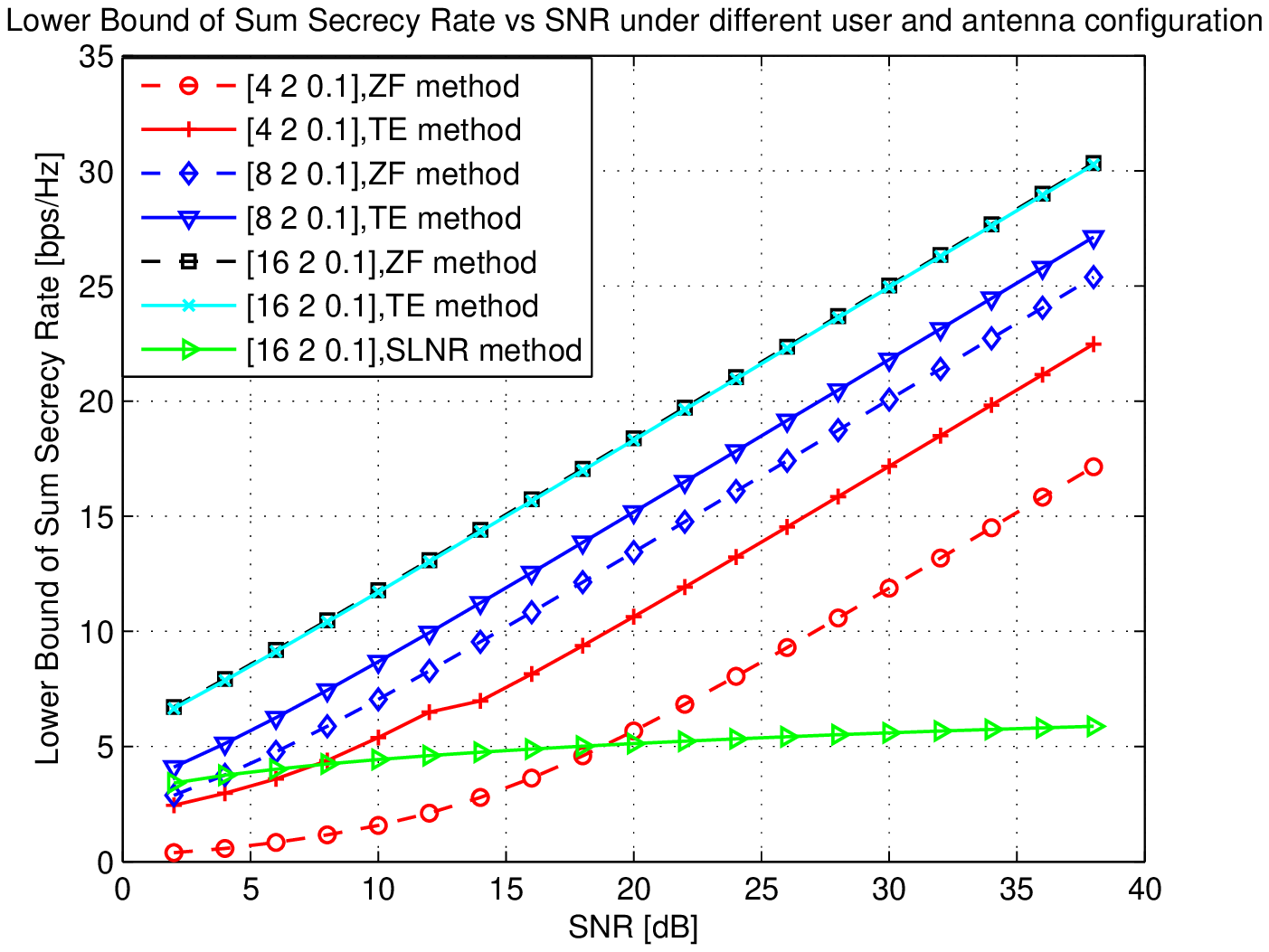}
\caption{lower bound of sum secrecy rate based on Taylor expansion method for $\varepsilon=0.1$ and $K=2$.}
\label{fig_2user01}
\end{figure}

\begin{figure}[!htbp]
\centering
\includegraphics[scale=0.57]{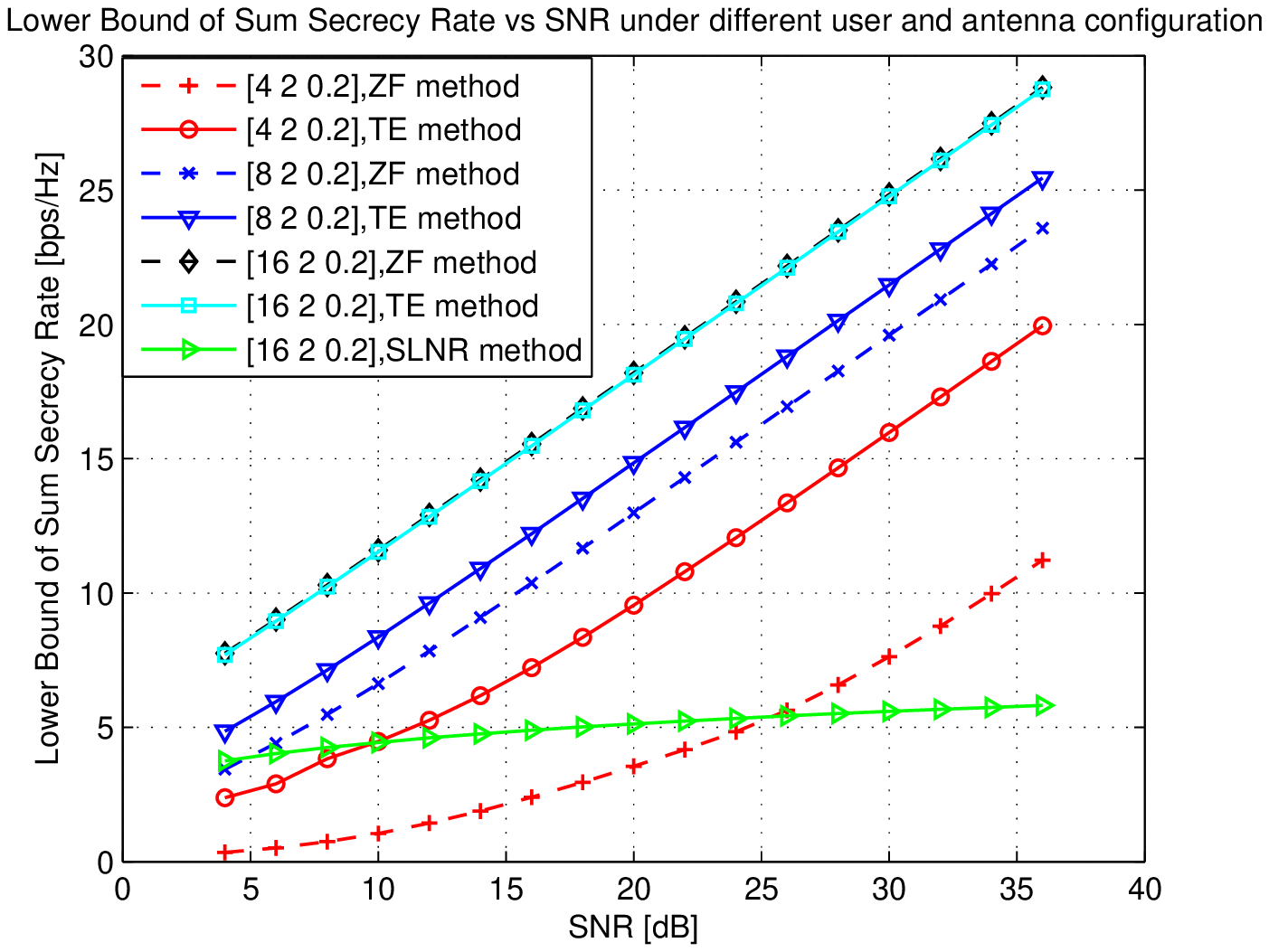}
\caption{lower bound of sum secrecy rate based on Taylor expansion method for $\varepsilon=0.2$ and $K=2$.}
\label{fig_2user02}
\end{figure}

\begin{figure}[!htbp]
\centering
\includegraphics[scale=0.57]{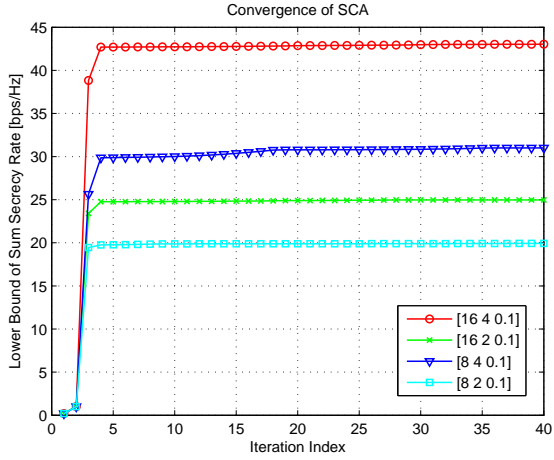}
\caption{convergence of SCA for $\varepsilon=0.1$.}
\label{fig_ite}
\end{figure}

\begin{figure}[!htbp]
\centering
\includegraphics[scale=0.57]{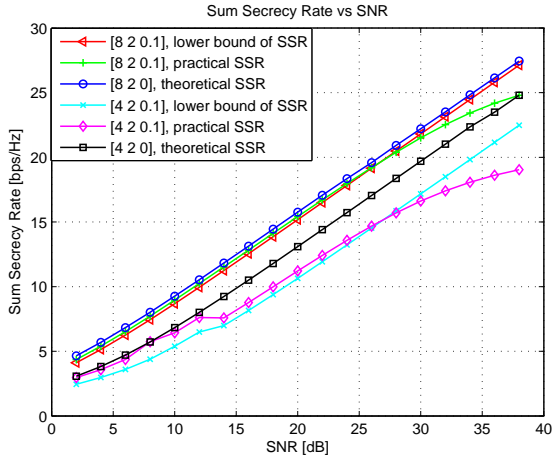}
\caption{comparison between the lower bound of SSR, the practical SSR and the theoretical SSR.}
\label{fig_real}
\end{figure}

\begin{figure}[!htbp]
\centering
\includegraphics[scale=0.57]{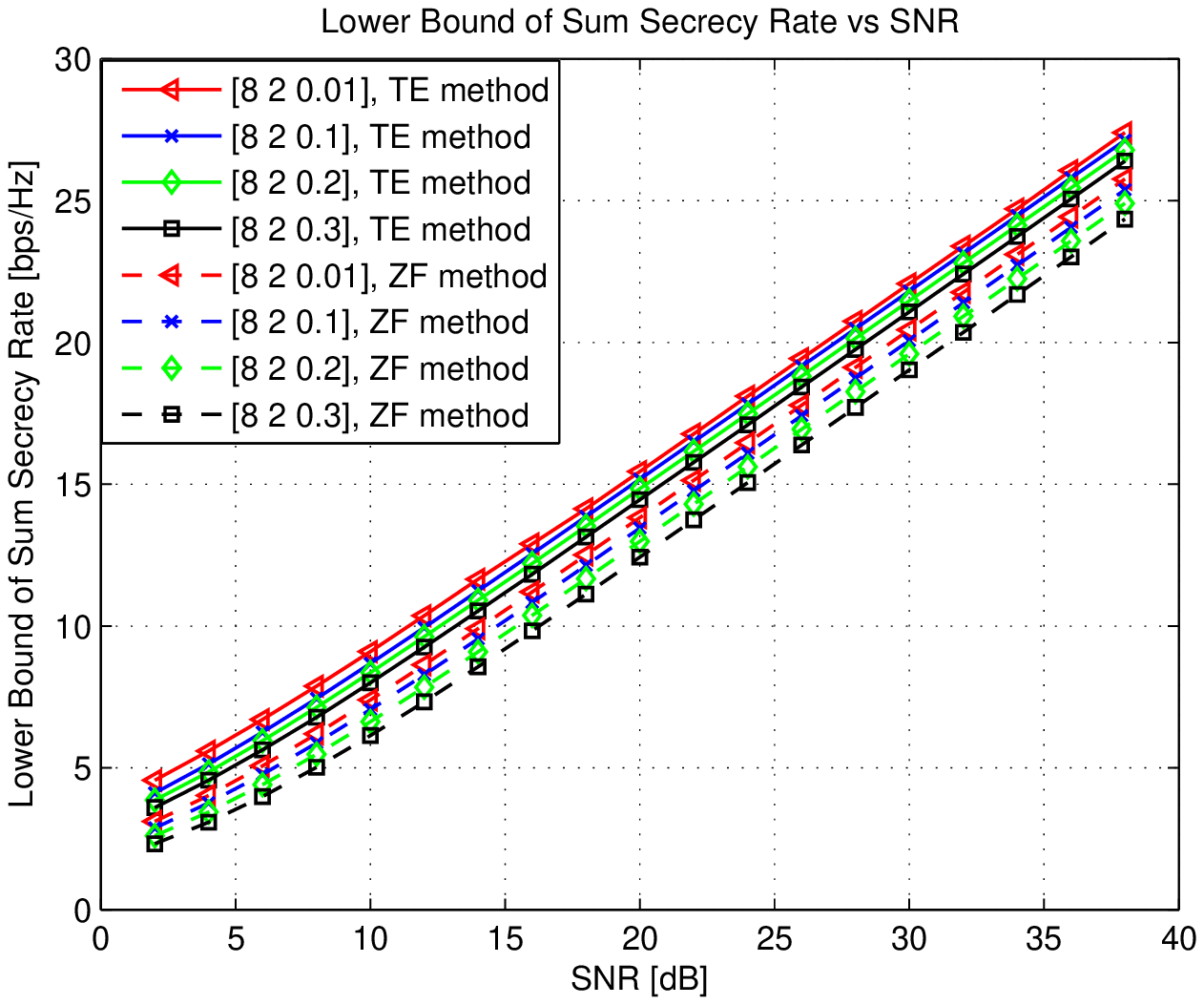}
\caption{comparison between different bounds of the norm of the channel estimation error.}
\label{fig_derr}
\end{figure}

\begin{figure}[!htbp]
\centering
\includegraphics[scale=0.57]{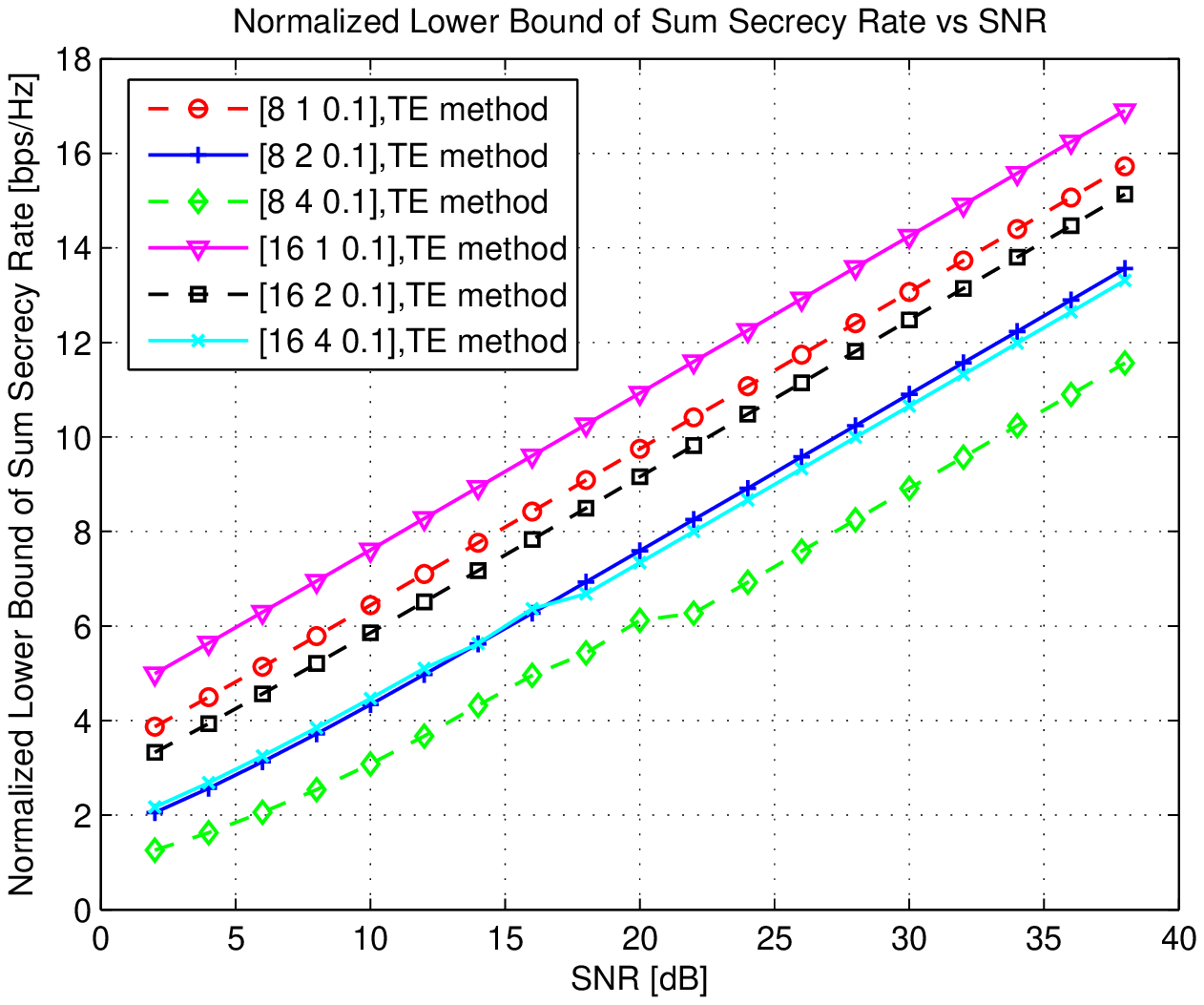}
\caption{comparison between different bounds of the norm of the channel estimation error.}
\label{fig_derr}
\end{figure}

\section{Conclusion}
In this paper, we have presented two efficient algorithms for solving the sum power constrained beamforming design problem for the purpose of maximum SSR. The two methods are mainly based on Taylor expansion and ZF algorithms, respectively. Numerical results show that the Taylor expansion based algorithm achieves better performance than the ZF based algorithm, which is mainly due to the freedom to optimize beamforming direction. In addition, both algorithms outperform the traditional SLNR algorithm.

\end{document}